\documentclass[12pt, draftclsnofoot, journal, letterpaper, onecolumn]{IEEEtran}
\DeclareSymbolFont{CMlargesymbols}{OMX}{cmex}{m}{n}
\DeclareMathSymbol{\sum}{\mathop}{CMlargesymbols}{"50}
\DeclareMathSymbol{\prod}{\mathop}{CMlargesymbols}{"51}

\DeclareMathAlphabet{\mathcal}{OMS}{cmsy}{m}{n}

\DeclareSymbolFont{Letters}{OML}{cmm}{m}{it}
\DeclareMathSymbol{\alpha}{\mathalpha}{Letters}{11}
\DeclareMathSymbol{\delta}{\mathalpha}{Letters}{14}
\DeclareMathSymbol{\epsilon}{\mathalpha}{Letters}{15}
\DeclareMathSymbol{\lambda}{\mathalpha}{Letters}{21}
\DeclareMathSymbol{\Lambda}{\mathalpha}{Letters}{3}
\DeclareMathSymbol{\pi}{\mathalpha}{Letters}{25}
\DeclareMathSymbol{\rho}{\mathalpha}{Letters}{26}
\DeclareMathSymbol{\sigma}{\mathalpha}{Letters}{27}
\DeclareMathSymbol{\Delta}{\mathalpha}{Letters}{1}
\DeclareMathSymbol{\Psi}{\mathalpha}{Letters}{9}

\DeclareMathSymbol{Q}{\mathalpha}{Letters}{81}

\DeclareSymbolFont{symbols}{OMS}{cmm}{m}{n}
\DeclareMathSymbol{\infty}{\mathord}{symbols}{"31}

\usepackage{scalerel,stackengine}
\stackMath
\newcommand\wihat[1]{%
\savestack{\tmpbox}{\stretchto{%
		\scaleto{%
			\scalerel*[\widthof{\ensuremath{#1}}]{\kern-.6pt\bigwedge\kern-.6pt}%
			{\rule[-\textheight/2]{1ex}{\textheight}}
		}{\textheight}%
	}{0.5ex}}%
\stackon[1pt]{#1}{\tmpbox}%
}

\usepackage[mathscr]{eucal}

\usepackage{eufrak}

\usepackage{lipsum} 
\usepackage[noadjust]{cite}
\usepackage{yfonts}
\usepackage{empheq}
\usepackage{algorithm}
\usepackage{algpseudocode}
\usepackage{amsthm}
\usepackage{amsfonts}
\usepackage{amssymb}
\usepackage{amsmath}
\usepackage{tabularx}
\usepackage{comment}
\usepackage{graphicx}
\usepackage[caption=false]{subfig}
\usepackage{caption}
\usepackage[dvipsnames]{xcolor}
\usepackage{mathtools}
\usepackage{float}
\usepackage{braket}
\usepackage{hyperref}
\hypersetup
{
	bookmarksdepth=-2,
	linkcolor=black,
	colorlinks=true,
	citecolor=blue
}

\usepackage[most]{tcolorbox}

\usepackage{cleveref}

\usepackage[acronym,nomain,nopostdot,nogroupskip]{glossaries}
\glossarystyle{listdotted}
\makeglossaries
\renewcommand*{\CustomAcronymFields}
{
	name={\the\glsshorttok},
	description={\the\glslongtok},
	first={\noexpand\emph{\the\glslongtok}\space(\the\glsshorttok)},%
	firstplural={\noexpand\emph{\the\glslongtok\noexpand\acrpluralsuffix}\space(\the\glsshorttok)},%
	text={\the\glsshorttok},%
	plural={\the\glsshorttok\noexpand\acrpluralsuffix}%
}

\newacronym{osp}{OSP}{Observer Selection Problem}

\newacronym{psd}{PSD}{positive semi-definite}

\newacronym{rl}{RL}{reinforcement learning}

\newacronym{lti}{LTI}{linear time-invariant}

\newacronym{zoh}{ZOH}{zero-order hold}
\newacronym{plc}{PLC}{programmable logic controller}

\newacronym{fcfs}{FCFS}{first-come first-serve}
\newacronym{fifo}{FIFO}{first-in first-out}
\newacronym{lut}{LUT}{lookup table}
\newacronym{dp}{DP}{dynamic programming}
\newacronym{bnb}{B\&B}{branch-and-bound}
\newacronym{kf}{KF}{Kalman filter}
\newacronym{lqg}{LQG}{linear-quadratic-Gaussian}
\newacronym{mdp}{MDP}{Markov decision process}
\newacronym{pomdp}{POMDP}{partially observable Markov decision process}

\newacronym{cdf}{CDF}{cumulative density function}
\newacronym{ccdf}{CCDF}{complementary cumulative density function}
\newacronym{pdf}{PDF}{probability density function}
\newacronym{pmf}{PMF}{probability mass function}
\newacronym{rhs}{RHS}{right hand side}
\newacronym{lhs}{LHS}{left hand side}

\newacronym{mse}{MSE}{mean squared error}
\newacronym{mmse}{MMSE}{minimum mean squared error}

\newacronym{snr}{SNR}{signal-to-noise ratio}

\newacronym{qos}{QoS}{QoS}
\newacronym{qoe}{QoE}{quality of experience}
\newacronym{lan}{LAN}{local-area network}
\newacronym{wan}{WAN}{wide-area network}
\newacronym{urllc}{URLLC}{Ultra-Reliable Low-Latency Communications}
\newacronym{embb}{eMBB}{enhanced Mobile Broadband}
\newacronym{pcf}{PCF}{point coordination function}
\newacronym{ncs}{NCS}{networked control system}

\newacronym{ap}{AP}{access point}
\newacronym{comp}{CoMP}{Coordinated Multipoint}
\newacronym{sumiso}{SU-MISO}{single-user multiple-input-single-output}
\newacronym{ibl}{IBL}{infinite block length}
\newacronym{fbl}{FBL}{finite block length}
\newacronym{icn}{ICN}{industrial control network}
\newacronym{wsn}{WSN}{wireless sensor network}
\newacronym{rt}{RT}{real-time}
\newacronym{tdm}{TDM}{time division multeplxing}
\newacronym{isi}{ISI}{inter-symbol interference}
\newacronym{nist}{NIST}{National Institute of Standards and Technology}
\newacronym{cbrs}{CBRS}{Citizens Broadband Radio Service}
\newacronym{los}{LOS}{line-of-sight}
\newacronym{nlos}{NLOS}{non-line-of-sight}
\newacronym{itu}{ITU}{International Telecommunications Union}
\newacronym{mmwave}{mmWave}{millimeter wave}
\newacronym{nsr}{NSR}{noise-to-signal ratio}
\newacronym{das}{DAS}{distributed antenna system}

\newacronym{csi}{CSI}{channel state information}
\newacronym{cqi}{CQI}{channel quality indicator}
\newacronym{ack}{ACK}{acknowledgement}
\newacronym{arq}{ARQ}{automatic repeat request}
\newacronym{awgn}{AWGN}{additive white Gaussian noise}
\newacronym{cc}{CC}{chase combining}
\newacronym{fec}{FEC}{forward error correction}
\newacronym{harq}{HARQ}{hybrid automatic repeat request}
\newacronym{hspa}{HSPA}{high speed packet access}
\newacronym{iid}{i.i.d.}{independent and identically distributed}
\newacronym{ir}{IR}{incremental redundancy}
\newacronym{lte}{LTE}{long term evolution}
\newacronym{mrc}{MRC}{maximal-ratio combining}
\newacronym{nack}{NAK}{negative acknowledgement}
\newacronym{wimax}{WiMax}{worldwide interoperability for microwave access}
\newacronym{3gpp}{3GPP}{3rd generation partnership project}
\newacronym{ofdm}{OFDM}{orthogonal frequency-division multiplexing}
\newacronym{ofdma}{OFDMA}{orthogonal frequency-division multiple access}
\newacronym{wlan}{WLAN}{wireless local area network}
\newacronym{gsm}{GSM}{global system for mobile communications}
\newacronym{edge}{EDGE}{enhanced data \gls{gsm} environment}
\newacronym{stbc}{STBC}{space-time block code}
\newacronym{amc}{AMC}{adaptive modulation and coding}
\newacronym{sinr}{SINR}{signal to interference and noise ratio}
\newacronym{mi}{MI}{mutual information}
\newacronym{acmi}{ACMI}{accumulated mutual information}
\newacronym{nacmi}{NACMI}{normalized ACMI}
\newacronym{cdi}{CDI}{channel distribution information}
\newacronym{latr}{LATR}{long-term average transmission rate}
\newacronym{rtr}{RTR}{round transmission rate}
\newacronym{fd}{FD}{full-duplex}
\newacronym{hd}{HD}{half-duplex}
\newacronym{td}{TD}{Time Division}
\newacronym{tdma}{TDMA}{time division multiple access}
\newacronym{mac}{MAC}{Media Access Control}
\newacronym{uwb}{UWB}{Ultra Wideband}
\newacronym{ieee}{IEEE}{institute of electrical and electronics engineers}
\newacronym{dB}{dB}{decibel}
\newacronym{min}{Min.}{minimum}
\newacronym{med}{Med.}{median}
\newacronym{avg}{Avg.}{average}
\newacronym{ul}{UL}{up-link}
\newacronym{dl}{DL}{downlink}
\newacronym{app}{APP}{a-posteriori probability}
\newacronym{logmap}{LogMAP}{log maximum a-posteriori}
\newacronym{llr}{LLR}{log-likelihood ratio}
\newacronym{ue}{UE}{user equipment}
\newacronym{5g}{5G}{5\textsuperscript{th} generation mobile networks}
\newacronym{4g}{4G}{4\textsuperscript{th} generation mobile networks}
\newacronym{tti}{TTI}{transmission time interval}
\newacronym{rrm}{RRM}{radio resource management}
\newacronym{mmib}{MMIB}{mean mutual information per bit}
\newacronym{dsi}{DSI}{decoder state information}
\newacronym{tb}{TB}{transport block}
\newacronym{tbs}{TBS}{transport block size}
\newacronym{cb}{CB}{code block}
\newacronym{cbg}{CBG}{code block group}
\newacronym{cbs}{CBS}{code block size}
\newacronym{prb}{PRB}{physical resource block}
\newacronym{rb}{RB}{resource block}
\newacronym{bler}{BLER}{block error rate}
\newacronym{blep}{BLEP}{block error probability}
\newacronym{crc}{CRC}{cyclic redundancy check}
\newacronym{tdd}{TDD}{time division duplexing}
\newacronym{fdd}{FDD}{frequency division duplex}
\newacronym{mcc}{MCC}{mission critical communication}
\newacronym{mmc}{MMC}{massive machine communication}
\newacronym{mtc}{MTC}{machine type of communication}
\newacronym{mmtc}{mMTC}{massive machine type of communication}
\newacronym{umtc}{uMTC}{ultra-reliable \gls{mtc}}
\newacronym{rtt}{RTT}{round trip time}
\newacronym{rs}{RS}{reference symbols}
\newacronym{kpi}{KPI}{key performance indicator}
\newacronym{kpis}{KPIs}{key performance indicators}
\newacronym{tx}{Tx}{transmitter node}
\newacronym{rx}{Rx}{receiver node}
\newacronym{cran}{C-RAN}{centralized radio access network}
\newacronym{rru}{RRU}{remote radio unit}
\newacronym{bbu}{BBU}{baseband unit}
\newacronym{fhd}{FHD}{fronthaul delay}
\newacronym{cch}{CCH}{control channel}
\newacronym{saw}{SAW}{stop-and-wait}
\newacronym{qci}{QCI}{\gls{qos} class identifier}
\newacronym{gbr}{GBR}{guaranteed bit rate}
\newacronym{mbr}{MBR}{maximum bit rate}
\newacronym{ngbr}{non-GBR}{non-\gls{gbr}}
\newacronym{arp}{ARP}{allocation and retention priority}
\newacronym{effcr}{ECR}{effective coding rate}
\newacronym{mcs}{MCS}{modulation and coding scheme}
\newacronym{eva}{EVA}{extended vehicular A}
\newacronym{epa}{EPA}{extended pedestrian A}
\newacronym{etu}{ETU}{extended typical urban}
\newacronym{re}{RE}{resource element}
\newacronym{reS}{REs}{resource elements}
\newacronym{nr}{NR}{new radio}
\newacronym{qpsk}{QPSK}{quadrature phase shift keying}
\newacronym{qam}{QAM}{quadrature amplitude modulation}
\newacronym{siso}{SISO}{single-input and single-output}
\newacronym{miso}{MISO}{multiple-input single-output}
\newacronym{mimo}{MIMO}{multiple-input multiple-output}
\newacronym{bs}{BS}{base station}
\newacronym{phy}{PHY}{physical layer}
\newacronym{rlc}{RLC}{radio link control}
\newacronym{bcfsaw}{BCF-SAW}{BCF-SAW}
\newacronym{bcf}{BCF}{backwards composite feedback}
\newacronym{bac}{BAC}{binary asymmetric channel}
\newacronym{bsc}{BSC}{binary symmetric channel}
\newacronym{dtx}{DTX}{discontinued transmission}
\newacronym{bpsk}{BPSK}{binary phase shift keying}
\newacronym{bep}{BEP}{bit error probability}
\newacronym{ndi}{NDI}{new data indicator}
\newacronym{dci}{DCI}{downlink control information}
\newacronym{csit}{CSIT}{channel state information at the transmitter}
\newacronym{lt}{LT}{loudest talker}
\newacronym{ct}{CT}{cooperative transmission}
\newacronym{bpcu}{bpcu}{bits per channel use}

\usepackage{tikz}
\usepackage{forest}
\usepackage{verbatim}
\usetikzlibrary{shapes,arrows,fit,calc,positioning,patterns,math}

\definecolor{clr_my_dgray}{rgb}{0.4, 0.4, 0.4}
\definecolor{clr_steel_yello}{RGB}{239, 174, 24}
\definecolor{clr_steel_blue}{RGB}{0, 82, 156}
\definecolor{clr_steel_red}{RGB}{199, 3, 45}

\makeatletter
\def\@IEEEsectpunct{:\ \,}
\def\paragraph{\@startsection{paragraph}{4}{\z@}{1.5ex plus 1.5ex minus 0.5ex}%
{0ex}{\normalfont\normalsize\bfseries}}
\makeatother

\makeatletter
\newsavebox{\measure@tikzpicture}
\NewEnviron{scaletikzpicturetowidth}[1]{%
	\def\tikz@width{#1}%
	\def\tikzscale{1}\begin{lrbox}{\measure@tikzpicture}%
		\BODY
	\end{lrbox}%
	\pgfmathparse{#1/\wd\measure@tikzpicture}%
	\edef\tikzscale{\pgfmathresult}%
	\BODY
}
\makeatother

\newcommand{\powerset}[1]{\mathscr{P}\left(#1\right)}

\DeclarePairedDelimiter\ceilceil{\lceil}{\rceil}
\DeclarePairedDelimiter\floorfloor{\lfloor}{\rfloor}
\newcommand{\ceil}[1]{\ceilceil*{#1}}
\newcommand{\floor}[1]{\floorfloor*{#1}}
\DeclareMathOperator*{\argmin}{\arg\,\min}

\newcommand{\expectno}[1]{\mathbb{E}\,#1}

\DeclarePairedDelimiter{\norm}{\lVert}{\rVert}
\newcommand{\tr}[1]{\textrm{#1}}
\newcommand{\trno}{{\textrm{tr}}}

\newcommand{\diff}{\mathrm{d}}

\theoremstyle{definition} 

\newtheorem{proposition}{Proposition}

\newcolumntype{C}{>{\(\displaystyle}c<{\)}@{}} 
\newcolumntype{L}{>{\(\displaystyle}l<{\)}@{}} 
\newcolumntype{R}{>{\(\displaystyle}r<{\)}@{}}

\newcommand{\secref}[1]{Sec.~\ref{#1}}
\newcommand{\figref}[1]{Fig.~\ref{#1}}

\newcommand{\transp}[1]{{#1}^{\mathsf{T}}}
\newcommand{\transpp}{\mathsf{T}}

\newcommand{\sft}{\mathsf{T}}
\newcommand{\bx}{\mathbf{x}}
\newcommand{\by}{\mathbf{y}}
\newcommand{\bu}{\mathbf{u}}
\newcommand{\bv}{\mathbf{v}}
\newcommand{\bw}{\mathbf{w}}
\newcommand{\baa}{\mathbf{A}}
\newcommand{\bbb}{\mathbf{B}}
\newcommand{\bcc}{\mathbf{C}}
\newcommand{\bqq}{\mathbf{Q}}
\newcommand{\brr}{\mathbf{R}}

\newcommand{\va}{\mathbf{a}}
\newcommand{\vc}{\mathbf{c}}
\newcommand{\ve}{\mathbf{e}}

\newcommand{\vs}{\mathfrak{s}}
\newcommand{\vx}{\mathbf{x}}
\newcommand{\vy}{\mathbf{y}}
\newcommand{\vu}{\mathbf{u}}
\newcommand{\vv}{\mathbf{v}}
\newcommand{\vw}{\mathbf{w}}

\newcommand{\vzero}{\mathbf{0}}
\newcommand{\vxh}{\hat{\mathbf{x}}}
\newcommand{\mxa}{\mathbf{A}}
\newcommand{\mxb}{\mathbf{B}}
\newcommand{\mxc}{\mathbf{C}}

\newcommand{\mxq}{\mathbf{Q}}
\newcommand{\mxr}{\mathbf{R}}
\newcommand{\mxs}{\mathbf{S}}
\newcommand{\mxp}{\mathbf{P}}
\newcommand{\mxi}{\mathbf{I}}
\newcommand{\mxphi}{\boldsymbol{\Phi}}
\newcommand{\mxlambda}{\mathbf{\Lambda}}
\newcommand{\vnu}{\boldsymbol{\nu}}

\newcommand{\myy}[3]{y_{#1}^{(#2)}[#3]}
\newcommand{\myt}[3]{t_{#1}^{(#2)}[#3]}
\renewcommand{\tr}[1]{\text{tr}\left(#1\right)}
\newcommand{\mse}{\mathsf{MSE}}

\begin{document}

\title{Scheduling Observers Over a Shared Channel with Hard Delivery Deadlines}
\author{Rebal Jurdi, Jeffrey G. Andrews, and Robert W. Heath Jr.\\
\thanks{Rebal Jurdi is currently with Samsung Research America, Jeffrey G. Andrews, and Robert W. Heath Jr. are with the Wireless Networking and Communications Group, The University of Texas at Austin.}}
\maketitle
\begin{abstract}
We abstract the core logical functions from applications that require ultra-low-latency wireless communications to provide a novel definition for reliability. Real-time applications — such as intelligent transportation, remote surgery, and industrial automation — involve a significant element of control and decision making.  Such systems involve three logical components: observers (e.g. sensors) measuring the state of an environment or dynamical system, a centralized executive (e.g. controller) deciding on the state, and agents (e.g. actuators) that implement the executive’s decisions. The executive harvests the observers' measurements and decides on the short-term trajectory of the system by instructing its agents to take appropriate actions.  All observation packets  (typically uplink) and action packets (typically downlink) must be delivered by hard deadlines to ensure the proper functioning of the controlled system.  In-full on-time delivery cannot be guaranteed in wireless systems due to inherent uncertainties in the channel such as fading and unpredictable interference; accordingly, the executive will have to drop some packets. 
We develop a novel framework to formulate the \textit{Observer Selection Problem} (OSP) through which the executive schedules a sequence of observations that maximize its knowledge about the current state of the system. To solve this problem efficiently yet optimally, we devise a branch-and-bound algorithm that systematically prunes the search space. 
Our work is different from existing work on real-time communications in that communication reliability is not conveyed by packet loss or error rate, but rather by the extent of the executive's knowledge about the state of the system it controls. 
\end{abstract}



\section{Introduction}\label{s:intro}
Emerging applications that require ultra low latency and high reliability and availability, e.g. autonomous driving, telemedicine and augmented reality \cite{simsek:2016}, are historically rooted in closed-loop process control and automation.
\Glsdesc{icn}s, also known as \glsdesc{ncs}s, are control systems where the control loops are closed through a communication network \cite{galloway:2013,moyne:2007}. These systems comprise of sensors, actuators and a centralized \gls{plc} that are interconnected by a wired network (known as \textit{fieldbus}) and interact as follows.
The sensors measure the state of an underlying plant or process and feed their measurements to the \gls{plc}.
The \gls{plc} computes the difference between the measured process variables and a desired setpoint, computes corrective actions, and sends those to the actuators \cite{frotzscher:2014}.
The actuators receive the actions and apply them, fulfilling the \gls{plc}’s control objective.
To guarantee synchronized behavior, all action packets must be delivered by a hard deadline \cite{neumann:2007}, which further induces an implicit deadline on delivering measurement packets.


The feedback loop linking sensors and actuators to the PLC is a template for a generic control or multistage decision process consisting of 3 key components: observers, agents, and a centralized executive. 
On a cycle-by-cycle basis,
the observers measure the state of an underlying dynamical system or environment and feed their observations to a centralized executive.
The executive processes the observations and decides on the short-term trajectory of the system by issuing a set of actions to
its agents who interact directly with the system.
In autonomous driving, a centralized motion planner can coordinate the motion of a fleet of vehicles by collecting their GPS and gyroscope readings and then controlling their throttles and brakes \cite{khabbaz:2019,badue:2019,pendleton:2017}. In telemedicine, a personal wellness system continuously monitors a patient's vital signs through a set of wearables for instant intervention and care provision \cite{kakria:2015,varshney:2007}.

Wires are impossible to deploy in many of these emerging applications due to their distributed nature.
The executive is typically a software application that runs on a mobile device, a server, or in the cloud; it integrates information from fragmented, possibly geographically-separated data sources \cite{grosky:2007}. 
The observers, agents, and centralized executive are not only logically distinct, but also physically distant entities that can only communicate wirelessly \cite{oteafy:2019}. 
Cutting the wires, however, introduces new challenges due to inherent uncertainties in wireless channels such as fading and unpredictable interference; consequently, the reliable and timely delivery of measurement and control packets can no longer be guaranteed.

The knowledge of the state of the system is the bedrock of the fulfillment of the executive's control policy, but wireless communications compromise this knowledge.
Based on the current state of the system, the executive transitions to a desired state by deciding on appropriate actions to be applied by its agents. There are, however, two uncertainties that prevent the perfect knowledge of state of the system: unpredictable interference and channel fading. 
Observations contaminated with unpredictable interference can prevent the inference of the true state of the system, and channel fades can prevent the exchange of all observation packets by their deadlines. 
{\color{black}
If there is insufficient time to transmit all observation packets, then only a select few should be transmitted. Accordingly, the executive has to identify the observations to be scheduled. 
In this paper, we assume that the executive selects the observations that maximize its knowledge of the state of the system so that it can instruct its agents to take commensurate actions to fulfill its policy.
Deciding which observations add most to the executive's knowledge may appear infeasible without knowledge of the content of these very observations. Given a model of a system's dynamics and its initial state, however, Bayesian estimation provides a methodology to evaluate state estimators from limited observations.
}

{\color{black}
The objective of the paper is to introduce and promote a notion of reliability that evaluates a communication system in relation to the application that it serves. The applications of interest are of the decision-making, control-type where a deadline is imposed on the delivery of exchanged packets. Information to be exchanged is abundant, yet communication resources are scarce, so packet loss will be rampant. Traditional metrics like packet loss rate and throughput cease to be appropriate as a measure of reliability because packet loss is inevitable. Hence, a new metric that evaluates the communication system is warranted. Since a communication system only exists to serve an application, it ought to be deemed as successful as said application (measured in an application-specific metric). We specify a common performance metric for real-time decision-making applications: the precision of the executive's knowledge about the state of the system it oversees, formalized in the Bayesian framework through the \gls{mse} of the state estimate. The reason for this choice is that when the executive has a reliable view of the system, it is able to take justified actions to control the system and implement its policy. Accordingly, the communication system is as reliable as the executive's view of the system. 
}

In this paper, we formulate the \textit{Observer Selection Problem} through which the executive schedules the set of observations that maximize its belief about the state of the system. We list the contributions of this paper next.

\subsection{Contributions}
The objective of this paper is to formulate and solve the Observer Selection Problem (OSP) by developing a framework that abstracts away the context around different control and multistage decision processes. Our contributions are:

\begin{itemize}
\item We derive the Kalman filter equations to predict the state of an \gls{lti} system at decision cycle boundaries from scalar observations sampled at different rates. We use an ensemble of observation models to describe multi-rate observations, and a continuous-time system description to tie together the observations and estimate the state of the system.

\item We formulate OSP which the executive solves every decision cycle for the \textit{optimal observation sequence}. We derive the objective function to be minimized, representing the \gls{mse} of the state estimate at decision cycle boundaries. We show that the \gls{mse} is an iterated function. We further define a latency and dependency constraint to ensure that the observations are collected by a hard deadline.
 
\item We design a \glsfirst{bnb} algorithm to solve OSP. The algorithm uses a \textit{subset forest} to systematically iterate over observation sequences by pruning a forest tree once it encounters a non-schedulable sequence.

\end{itemize}

\subsection{Related Work}
Our work lies at the intersection of three research areas: estimation in \glsdesc{ncs}s, sensor selection, deadline-constrained scheduling, and is related to age of information and resource allocation in 5G \gls{urllc}. For every research area, we define the central problem, frame the novelty of our work within related work, and highlight the key points of difference. 

\paragraph{Estimation in networked control systems} Research in this area solves control and estimation problems where sensor measurements and actuator controls are sent over lossy networks rife with packet delays and dropouts \cite{hespanha:2007,schenato:2007}. Our work is different than existing research in four ways. 
First, existing work analyzes the estimation \gls{mse} of a Kalman filter when measurement packets are randomly dropped in the network \cite{shi:2010,sinopoli:2004}, but our work considers the problem of \textit{selecting the set of observations} that minimize \gls{mse}; accordingly, some observations are deliberately dropped. 
Second, prior work assumes a uniform rate at which the system's state is updated, inputs refreshed, outputs sampled, and estimate updated, but we assume that the outputs are sampled at different rates, leading to an ensemble of observation models with distinct time references. While there is an existing niche in multi-sensor multi-rate estimation \cite{glasson:1980}, work therein resolves sampling rate incoherence and simplifies analysis by assuming that sampling rates are integer multiples of the state update rate \cite{fangfang:2014,liang:2009,yan:2007,yan:2006}. 
Third, prior work uses a continuous-time state description to bridge multi-rate observations in a distributed context, but we use this description in a centralized context. Most related to our work, \cite{alouani:2005} considers a distributed estimation problem and proposes a rule for synchronizing local \textit{estimates} produced by noncoherent sensors at a fusion center. Unlike prior work, we propose a rule for synchronizing local \textit{observations} at the (centralized) executive.
Fourth, previous work in multi-rate estimation propose rules for updating the state estimate from out-of-order measurements (including the so-called mixed-time filters) \cite{barshalom:2002, jiang:2017,yan:2015}, but we update the estimate from ordered measurements.

\paragraph{Sensor selection} The sensor selection problem is defined as selecting a subset of available sensors to optimize a utility function, such as localization accuracy, while constraining the number of activated sensors, their sum energy consumption, or their distance to a centralized executive \cite{rowaihy:2007,bian:2006,luo:1989}. Most related to our work, \cite{hashemi:2018,hashemi:2017}, tackle a Kalman filtering problem with partial observations. The two papers study the problem of sensor selection to minimize the \gls{mse} regularized with a function that rewards balanced performance of individual sensors under a cardinality constraint. Our work is different in two main aspects. First, \cite{hashemi:2018,hashemi:2017} constrain the number of active sensors, which is a surrogate for constraining the total bandwidth required to transmit the sensor observations. Our work, however, sets a hard deadline for delivering the observations. Second, \cite{hashemi:2018,hashemi:2017} applies the standard Kalman filter formulation, whereas we apply a multirate formulation that is specifically tailored to the activation model.

\paragraph{Deadline-constrained scheduling} The deadline-constrained scheduling problem is defined as follows. A scheduler is presented with a combination of periodic and sporadic tasks and their interarrival times, execution times and deadlines, and its objective is to determine a time schedule for task execution so that no deadline miss occurs \cite{langer:2017,davis:2011}. The main difference between this paper and the scheduling literature is as follows. In the scheduling problem, the scheduler knows the tasks it needs to schedule. In our problem, however, the scheduler selects an optimal set of tasks and proceeds to serve them on a \gls{fcfs} basis; tasks in our paper are identified with packet transmissions. Prior work on scheduling deadline-constrained periodically-generated packets over wireless networks determines optimal schedules based on different criteria. For example, \cite{becchetti:2009} minimizes the energy consumption of transmitting packets, \cite{adamou:2001} maximizes network throughput, and \cite{hariharan:2011} maximizes the aggregated information. Our work, however, selects optimal tasks by minimizing the estimation \gls{mse}, a non-standard measure of service quality of a network.

{\color{black}
\paragraph{Age of Information} Introduced in \cite{kaul:2012,kaul:2012:2}, the \textit{age of information} (AoI) measures the freshness of real-time status updates of an evolving system or process. It is defined as the difference between the current time and the timestamp of the latest status update. This area is similar to our work in the setup. An observer measures and timestamps a state of interest (e.g. velocity, acceleration) and transmits it to the ``agency that monitors the system'' (executive) whose goal is to predict and control the system. Work on AoI, however, is vastly different from this work in that it uses the classical queueing theoretic framework to define the AoI metric. The central problem in work on AoI is to analyze and optimize (minimize) AoI for queueing systems with different arrival and departure process, service times, buffer capacities, number of servers, and service disciplines. On the contrary, we do not use the queueing framework but build our own instead.  Also, the problem that we address is minimizing the \gls{mse} of the state by selecting the optimal observation sequence. While recent work on AoI \cite{costa:2016,kam:2018} has introduced the notion of packet deadlines, the definition is quite different than ours. A deadline in this paper is an absolute time by which a packet needs to be delivered. A deadline according to \cite{costa:2016,kam:2018} is a timer per which a packet is dropped when it times out.
}

\paragraph{Resource allocation in \gls{urllc}} Our work is fundamentally different from existing work on resource allocation for low-latency communication services, such as 5G's \gls{urllc}, in the following way. We regard \gls{urllc} as catering to mission-critical applications that adhere to the structure of control and multistage decision processes. Recent work on \gls{urllc} studies resource allocation for downlink \gls{urllc} traffic \cite{anand:2018,she:2018:2}, for uplink and downlink \cite{she:2018}, and for \gls{urllc} and \gls{embb} traffic \cite{anand:2018:2,popovski:2018}. More recent work \cite{hou:2020,she:2020} has demonstrated that communications and control systems can be co-designed to reduce user-experienced delay in URLLC services. The scheduler allocates time-frequency resources to meet \glsdisp{qos}{QoS} requirements, e.g. outage probabilities and packet error rates. In our paper, however, these targets are irrelevant, and they are neither the object of minimization nor a constraint thereon. At the highest level, an application's success depends on it's fulfillment of the executive's control and decision making objectives. Accordingly, the reliability of the network that supports the application should not be conveyed by packet loss or error rate, but rather by the executive's knowledge about the state of the system that it controls.

\subsection{Organization and Notation}
The remainder of the paper is organized as follows. In \secref{s:b1}, the observation model is introduced. In \secref{s:b2}, the control system model is introduced, and the corresponding state transition and Kalman filter equations are derived. In \secref{s:b3}, the \gls{osp} is formulated. In \secref{s:b4} a solution algorithm is devised. In \secref{s:num}, key numerical results are presented.

In this paper, tuples are used extensively. The notation $\vs$ is used to denote a tuple of an arbitrary length. Subscripts are used to pick out elements of a tuple. For example, $s_\ell$ is the $\ell$th element of $\vs$, $\ell\geq 0$. When the length of $\vs$ is ambiguous, negative subscripts are used to pick out elements in reverse order. For example, $s_{-\ell}$ is the $\ell$th element from the end of $\vs$. As a tuple is also an ordered set, the length of $\vs$ is denoted by $|\vs|$, and membership is expressed as $u\in\vs$ if $u=s_\ell$ for some $\ell$. When the length of $\vs$ is ambiguous, it is  convenient at times to denote by $s_{|\vs|}$ the last element of $\vs$. The powerset of $\vs$, the set containing all subsequences of $\vs$, is $\powerset{\vs}$. Prepending and appending $u$ to $\vs$ is denoted by $(u,\vs)$ and $(\vs,u)$. Additionally, the following notation is used. Bold uppercase $\mxa$ denotes, bold lowercase $\va$ denotes a column vector, and non-bold lowercase $a$ and uppercase $A$ denote scalar quantities. Further, $\norm{\va}$ is the $\ell_2$ (Euclidean) norm of $\va$, and  $\transp{\va}$ is its transpose. The identity matrix is denoted by $\mxi$. The space of positive semi-definite matrices of dimension $L$ is denoted by $\mathbb{S}_+^L$. Finally, $\delta(t)$ is the Dirac delta function.

\section{System Model}\label{s:b1}
Having introduced the observer-executive-agent abstraction of control and decision processes, we model the relationship between the observations, actions, and internal state of the system. In the subsequent section, we use this model to derive state estimators from different observation sequences and compare their precision.

\begin{figure*}
\centering
\includegraphics[width=\textwidth]{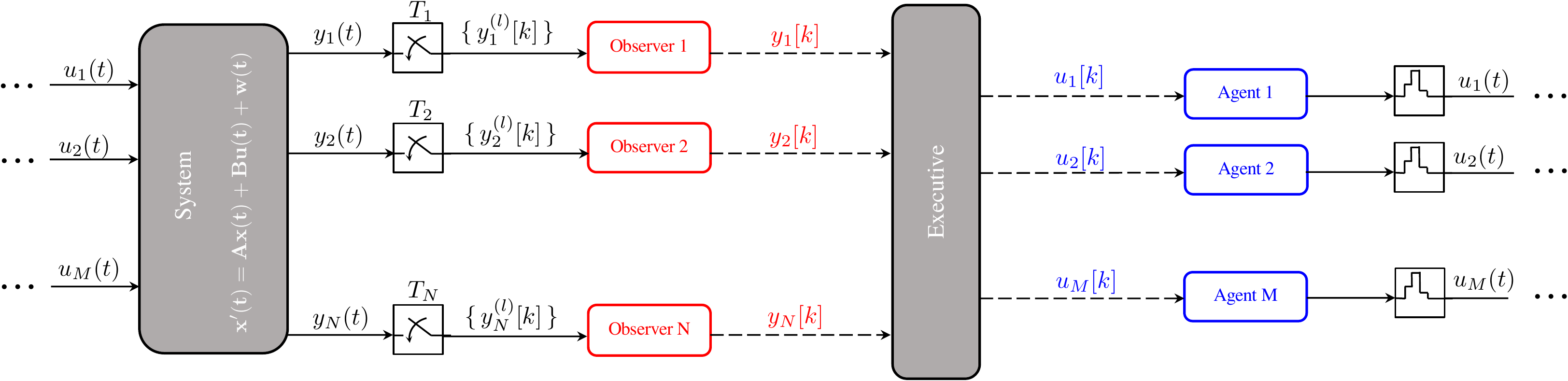}\label{fig:decision_loop}
\caption{The 3 components of a decision application: observers, agents, and an executive. These 3 entities communicate over a shared wireless medium. The executive observes the system through a set of observers and drives it through a set of agents.}
\end{figure*}

{\color{black}
\subsection{The Executive: Decision Maker and Scheduler}
The mathematical framework that will be developed in this paper fits any cellular or standard infrastructure-based wireless network where the executive oversees the scheduler, and any proprietary networked control system where the executive can be designed to assume any role. As such, the executive will fulfill the following roles: decision making (system controller) and scheduling (medium access controller). The decision making arm of the executive will run as an application layer process on a local server and assume full command over the executive's scheduling arm. The decision maker will instruct the scheduler to clear the medium for observers to send their representative observations one after the other. The scheduler will receive observation packets and pass them up the stack to the decision maker that will then determine the appropriate actions to be taken by the agents. It will encapsulate these actions into packets and pass them down to the scheduler to send them to the agents one after the other.
}

\subsection{Decision Cycle}
The executive performs four functions in chronological order in every decision cycle. First, it selects a subset of observations to harvest and determines a time schedule for the exchange of observation and action packets. Second, it harvests the observations, one observer at a time. Third, it uses these observations to estimate the current state of the system, decides on the future state and determines corrective actions to transition to that state. Fourth, it dispatches the actions, one agent at a time. We assume that all agents apply their actions \textit{isochronously} at decision cycle boundaries, i.e. simultaneously and at regular intervals. To apply new actions every cycle, the agents must receive these actions by the start of every cycle. As different observers may obtain their measurements at different scales or resolutions, we assume they sample the output of the system periodically but nonuniformly. Harvesting all observations might not leave enough time in the decision cycle to determine and dispatch actions. In that scenario, the executive can afford to harvest only a subset of the observations. Accordingly, we assume that the executive opts to select the set of observations that maximize its knowledge about the current state of the system. Technically, this is equivalent to minimizing the error between the true state and that perceived by the executive.

\subsection{Dynamical System}
Markov chains and differential (continuous-time) and difference (discrete-time) equations are the most common models to describe the time evolution of dynamical systems \cite{darling:2008}. While a Markov chain defines how a system evolves organically, a \gls{mdp} includes an \textit{agent}\footnote{In this paper, observers, executives, and agents refer to entities performing logically-distinct functions. In the \gls{mdp} context, however, these functions are performed by a single entity: the agent.} that regularly decides on the system trajectory based on its history \cite{puterman:book}. In this paper, we describe the dynamical system in state space form, which is a set of first-order differential/difference equations relating the system state to its inputs and outputs. This allows us to pose the problem of state estimation from streaming observations as a Kalman filtering problem which is more tractable. We defer the \gls{mdp} model to future work.


The executive observes the system through $N$ observers and interacts with it through $M$ agents on a cycle-by-cycle basis. We define $\vx(t)$, $\vu(t)$ and $\vy(t)$ to be the state, input and output of the system at time $t$, and $y_n(t)$ the scalar output that is observable by observer $n$, $1\leq n\leq N$. 
We consider a \gls{lti} continuous-time dynamical system represented in state-space form as
\begin{align}
\label{eq:ss_state}
\bx'(t) &= \baa \bx(t) + \bbb \bu(t) + \bv(t), \\
\label{eq:ss_output}
\by(t)  &= \bcc \bx(t) + \bw(t),
\end{align}
where $\displaystyle \bx'(t)\triangleq \diff \bx(t)/\diff t$, $\baa\in\mathbb{R}^{S\times S}$ is the state transition matrix, $\bbb\in\mathbb{R}^{S\times M}$ is the control/input model, $\bcc\in\mathbb{R}^{N\times S}$ is the observation/measurement matrix, $\bv(t)\in\mathbb{R}^{S}$ is the white Gaussian process uncertainty with covariance $\bqq(t)\delta(t)$, and $\bw(t)\in\mathbb{R}^{N}$ is the white Gaussian observation noise with covariance $\brr(t)\delta(t)$.

\subsection{Periodic Observations and Actions}
We define $T$ be the \textit{decision period}, i.e. the period of the decision cycle during which the executive must complete all 4 tasks. We define $T_n$ to be the \textit{observation period} of observer $n$, i.e. the period at which observer $n$ samples $y_n(t)$, the output observable by observer $n$. 
{\color{black}
In a decision cycle, every observer produces a sequence of one or more observations. These observations do not add information to one another in the way that a sequence of video packets together construct a video frame. Instead, they overtake one another like status updates: once the latest update is available, former updates become obsolete; hence, only the latest observation of an observer that can be transferred should be transferred.  
}

To make the notation more compact, we define the variable $\myy{n}{\ell}{k}$ to be observation $\ell$ of observer $n$ in decision cycle $k$, and the variable $\myt{n}{\ell}{k}$ to be the timestamp of this observation. The relationship between these two variables is
\begin{align}\label{eq:obs}
\myy{n}{\ell}{k}=y_n\left(\myt{n}{\ell}{k}\right). 
\end{align}

{\color{black}
The number of sequences the executive can use to form its belief about the state of the system grows exponentially with the number of observations available in the decision cycle. Allowing for multiple observations per observer per cycle, the total number of observation sequences is
\begin{align}
2^{\sum_{n=1}^N \ceil{\frac{T}{T_n}}}.
\end{align}
If every observer is limited to at most one observation per cycle, the number of observations is at most $2^N$ which is a significant reduction in the size of the search space of observation sequences. Indeed, we assume that every observer can report at most one \emph{representative observation} every decision cycle $k$. 
}
We denote by $y_n[k]$ and $t_n[k]$ the representative observation of observer $n$ in decision cycle $k$ and its timestamp. This observation could be defined systematically, arbitrarily, or randomly. An obvious choice for an observer's representative observation is its last observations in a decision cycle. On the one hand, the last observation reflects a state closer in time to that at the end of the decision cycle. On the other hand, sending that observation could leave little time to make a decision on the actions to be applied at the start of the next cycle. An alternative choice is an observer's first observation in a decision cycle. While the first observation reflects a state farther in time from that at the end of the decision cycle, sending this observation leaves ample time to make a decision and dispatch actions.
{\color{black}
Still, this does not fully justify harvesting an observation that reflects an early state of the system from which the current state may have sufficiently diverged. We shall assume that the state or process variables that are sampled by the observers do not change substantially over the decision cycle, i.e. the variables have a ``coherence time'' less than the decision period. If the process variables that are critical to decision making accuracy do vary rapidly within a decision cycle, then the choice of the decision period ought to be reconsidered. If this is not possible, then more bandwidth should be recruited to allow for on-time delivery of observations produced late into the cycle. 
}
With that, we assume that an observer's representative observation in a decision cycle is its first observation in that cycle. We group these representative observations for every cycle $k$ into the observation vector 
\begin{align}
\by[k] = \big[\; y_0[k] \;\; y_1[k] \;\; \cdots \;\; y_{N-1}[k] \;\big]^\sft.
\end{align}

%
%
%
%

\definecolor{myred}{RGB}{255, 102, 102}
\definecolor{mygray}{RGB}{122, 122, 122}
\definecolor{myblue}{RGB}{102, 179, 255}


\begin{figure*}[t]
\centering
\begin{tikzpicture}[xscale=1,transform shape]
\def\T{7};
\def\tau{3}

\draw [-latex'](-0.5,0) -- (0,0)coordinate(o) -- (2*\T+2,0)coordinate(tmax) node[above]{$t$};
\draw[thick] (0,0.2) -- (0,-0.2) node[below ]{$(k-1)T$};
\draw[thick] (\T,0.2) -- (\T,-0.2) node[below ]{$kT$};
\draw[thick] (2*\T,0.2) -- (2*\T,-0.2) node[below ]{$(k+1)T$};

\draw[domain=-0.5:2*\T+2,smooth,variable=\x,red,thick] plot ({\x},{2+4*exp(-0.2*\x)*cos(30*\x)});
\foreach \k in {0,1,...,5}
{
	\pgfmathsetmacro{\x}{0.3+\k*\tau};
	\coordinate  (f-\k) at (\x,{2+4*exp(-0.2*\x)*cos(30*\x)});
	\draw[-latex,thick,red] (f-\k) --(\x,0);
	\node at (f-\k)[color=red,circle,fill,inner sep=1.5pt]{};
}

\node at (0.3,5.7) [above right]{\color{red} $y_n^{(0)}[k]$};
\node at (3,1.9) [above right]{\color{red} $y_n^{(1)}[k]$};
\node at (5.5,1.0) [above right]{\color{red} $y_n^{(2)}[k]$};
\node at (7.9,2.2) [above right]{\color{red} $y_n^{(0)}[k+1]$};
\node at (11,2.4) [above right]{\color{red} $y_n^{(1)}[k+1]$};
\node at (14,2.1) [above right]{\color{red} $y_n^{(0)}[k+2]$};

\end{tikzpicture}	
\caption{The output of the system $y_n(t)$ that is observable by observer $n$. The observer periodically samples this output to obtain a sequence of observations $\set{\myy{n}{l}{j}}$.}
\end{figure*}
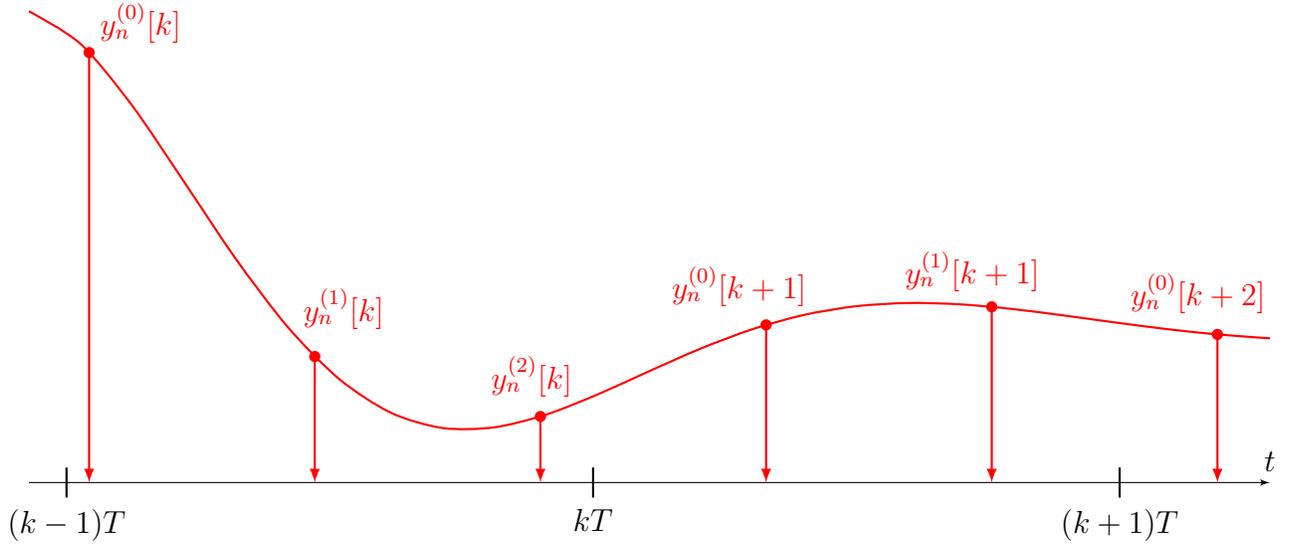




{\color{black}
An observer's sampling rate should be chosen by the executive to be commensurate with the observer's value. Too small, and the observer and its observations will be considered essential to building the executive's cycle-to-cycle belief of the state of the system. Too large, and the observer will be non-essential and might as well be dropped.
}
Comparing $T_n$ with $T$, we can determine the number of observations in decision cycle $kT$ as well as their timestamps. We distinguish between 3 cases, (a) $T_n=T$, (b) $T_n\leq T$, and (c) $T_n>T$. The following proposition gives the timestamp of the first observation for these three cases.

\begin{proposition}\label{p:observation_timestamps}
When $T_n=T$, the timestamp of the first observation in cycle $k$ is
\begin{align}
\myt{n}{0}{k} &= (k-1)T.
\intertext{When $T_n\leq T$,}
\myt{n}{0}{k} &= \ceil{\frac{(k-1)T}{T_n}} T_n.
\intertext{When $T_n>T$,}
\myt{n}{0}{k} &= \floor{\frac{kT}{T_n}} T_n
\intertext{if $kT - \floor{\frac{kT}{T_n}}\, T_n \leq T$. Otherwise, observer $n$ produces no observations that cycle.}
\end{align}
\end{proposition}
\begin{proof}
See Appendix~\ref{app:b}.
\end{proof}

The count, order, and identity of observations from any observer in any decision cycle can be determined offline. In \secref{s:b3}, we will see that an observation's timestamp plays a critical role in the executive's decision to select that observation.


We have described how the sampled observations for observer $n$, $\set{\myy{n}{\ell}{k}}$, are produced from the continuous-time system outputs $\set{y_n(t)}$. We now describe how the system inputs are constructed from the agents' actions. In real-time control systems, all of the system inputs are refreshed isochronously at the start of every cycle. Accordingly, we model the continuous-time inputs $\bu(t)$ as a \gls{zoh} of the discrete-time actions $\bu[k]$. The \gls{zoh} holds the values of the inputs steady for the entire duration of the decision cycle \cite{ogata:book}. Mathematically, $\bu(t)$ is reconstructed from $\bu[k]$ as
\begin{align}
u(t) = u\left[\floor{\frac{t}{T}}\right].
\end{align}

There are four key points to take away from this section. First, an observer produces a different number of observations in different decision cycles, which depends on the relative values of the observation and decision periods. Second, every observer can report at most one observation in a decision cycle, which is that observer's representative observation for that cycle. Third, the timestamp of every representative observation in any cycle can be computed offline. Fourth, while observers sample the output of the system at different rates, actions are applied isochronously at decision cycle boundaries.

\section{State Estimation}\label{s:b2}
In this section, we explain how the executive maximizes its belief about the system state from streaming observations and past actions. We first derive the state transition equations that relate the state of the system at two different times, and then we use the Bayesian framework to derive the optimal state estimator used by the executive, i.e. the Kalman filter.

\subsection{State Transition Equations}
The relationship between the state of the system at two different time stamps is key to state estimation at a time of interest from a previous state.
The actions that are applied at the decision cycle boundary (time $kT$) should account for the state of the system at that boundary. The reason is that the evolution of the system beyond time $kT$ is determined by the values of the state and input at time $kT$ (see \eqref{eq:ss_state}). The executive, however, can only use observations produced well ahead of decision cycle boundaries; it takes time to harvest and process observations and to decide on and dispatch actions. Accordingly, the executive needs to \emph{predict} the state at decision cycle boundaries. This prediction entails relating observations with arbitrary timestamps to the state at cycle boundaries. 
Let $t[k]$ be an arbitrary time in cycle $k$, i.e. $(k-1)T \leq t[k] \leq kT$. In the next proposition, we express the state at time $t[\ell]$ as a function of the state at an earlier time $t[k]$. 
\begin{proposition}[State transition equation]\label{p:state_transition_eq}
The state at an arbitrary time $t[\ell]$ is related to that at an earlier time $t[k]$ according to the following equation:
\begin{equation}\label{eq:state_update}
\begin{alignedat}{2}
\vx \Big(t[\ell]\Big)	
&= \mxphi\Big(t[k],t[\ell])\Big) \vx(t[k])
&&+ \mxlambda\Big(t[k],(k+1)T,t[\ell]\Big)\bu[k] \\
& &&+ \sum_{j=k+1}^{\ell-1} \mxlambda\Big(jT,(j+1)T,t[\ell]\Big)\bu[j] \\
& &&+ \mxlambda\Big(lT,t[\ell],t[\ell]\Big)\bu[\ell] + \vnu\Big(t[k],t[\ell]\Big), \\
\end{alignedat}
\end{equation}
where
\begin{align}
\mxphi(s,t) &= e^{\mxa(t-s)}, \\
\mxlambda(r,s,t) &= \mxa^{-1}e^{\mxa (t-s)}\left(e^{\mxa (s-r)}-\mxi\right) \mxb, \\
\vnu(s,t) &= \int_{s}^{t} e^{\mxa(t-\tau)} \vv(\tau) \diff\tau.
\end{align}
The matrix $\mxphi(s,t)$ is known as the state transition matrix between time instances $s$ and $t$, $s\leq t$.  The vector $\vnu(s,t)$ is Gaussian with a mean of $\vzero$ and a covariance
\begin{align}
\mxq\Big(t[k],t[\ell]\Big)
= \int\limits_{0}^{t[\ell]-t[k]} e^{\mxa s} \mxq e^{\mxa^* s} \diff s.
\end{align}
\end{proposition}
\begin{proof}
Similar to Chapter 5 Section 5 of \cite{ogata:book}.
\end{proof}

\begin{figure*}
\centering
\subfloat[Inter-cycle transition]{\label{fig:x2x_inter}
\begin{tikzpicture}[>=stealth']
\def\T{5}
\def\L{14}
\draw[->] (0,0) -- (\L,0);
\node at (\L,0) [label=below:$t$] {};

\coordinate (k1) at (0.5,0);
\coordinate (k2) at ($(k1) + (\T,0)$);
\coordinate (l1) at ($(k2) + (2,0)$);
\coordinate (l2) at ($(l1) + (\T,0)$);
\coordinate (t1) at ($(k1) + (1,0)$);
\coordinate (t2) at ($(l2) - (2.5,0)$);

\foreach \n in {k1,k2,l1,l2}
\draw ($(\n)-(0,0.2)$) -- ($(\n)+(0,0.2)$);
\node at (k1) [label=below left:$(k-1)T$] {};
\node at (k2) [label=below:$kT$] {};
\node at (l1) [label=below:$(\ell-1)T$] {};
\node at (l2) [label=below:$\ell T$] {};

\foreach \n in {t1,t2}
	\node at (\n)[blue,circle,fill,inner sep=1.5pt]{};
\node (1) at (t1) [label=below:{\color{blue} $t[k]$}] {};
\node (2) at (t2) [label=below:{\color{blue} $t[\ell]$}] {};
\draw[->,blue] (1) to [out=15,in=165] (2);

\end{tikzpicture}
}

\subfloat[Edge-to-edge transition]{\label{fig:x2x_e2e}
\begin{tikzpicture}[>=stealth']
\def\T{5}
\def\L{14}
\draw[->] (0,0) -- (\L,0);
\node at (\L,0) [label=below:$t$] {};

\coordinate (k1) at (0.5,0);
\coordinate (k2) at ($(k1) + (\T,0)$);
\draw ($(k1)-(0,0.2)$) -- ($(k1)+(0,0.2)$);
\draw ($(k2)-(0,0.2)$) -- ($(k2)+(0,0.2)$);
\node (1) at (k1) [label=below left: {\color{blue}  $(k-1)T$} ] {};
\node (2) at (k2) [label=below: {\color{blue}  $kT$} ] {};

\foreach \n in {k1,k2}
	\node at (\n)[blue,circle,fill,inner sep=1.5pt]{};
	
\draw[->,blue] (1) to [out=15,in=165] (2); 

\end{tikzpicture}
}

\subfloat[Intra-cycle transition]{\label{fig:x2x_intra}
\begin{tikzpicture}[>=stealth']
\def\T{5}
\def\L{14}
\draw[->] (0,0) -- (\L,0);
\node at (\L,0) [label=below:$t$] {};

\coordinate (k1) at (0.5,0);
\coordinate (k2) at ($(k1) + (\T,0)$);
\coordinate (t1) at ($(k1) + (0.5,0)$);
\coordinate (t2) at ($(k2) - (1.5,0)$);

\draw ($(k1)-(0,0.2)$) -- ($(k1)+(0,0.2)$);
\draw ($(k2)-(0,0.2)$) -- ($(k2)+(0,0.2)$);
\node at (k1) [label=below left:$(k-1)T$] {};
\node at (k2) [label=below:$kT$] {};

\foreach \n in {t1,t2}
	\node at (\n)[blue,circle,fill,inner sep=1.5pt]{};
\node (1) at (t1) [label=below: {\color{blue} $\myt{}{1}{k}$} ] {};
\node (2) at (t2) [label=below: {\color{blue} $\myt{}{2}{k}$} ] {};
\draw[->,blue] (1) to [out=15,in=165] (2); 

\end{tikzpicture}
}
\caption{The state of the system $\vx(t)$ at an arbitrary time as computed from an earlier state $\vx(s)$, $s<t$; 2 special cases are highlighted. In part (a), the state at a desired time instance is computed from an earlier state in an earlier cycle. In part (b), the state at the cycle's right edge is compute from the state at the cycle's left edge. In part (c), the state is computed from an earlier state in the same cycle.}
\label{fig:kalman_step}
\end{figure*}
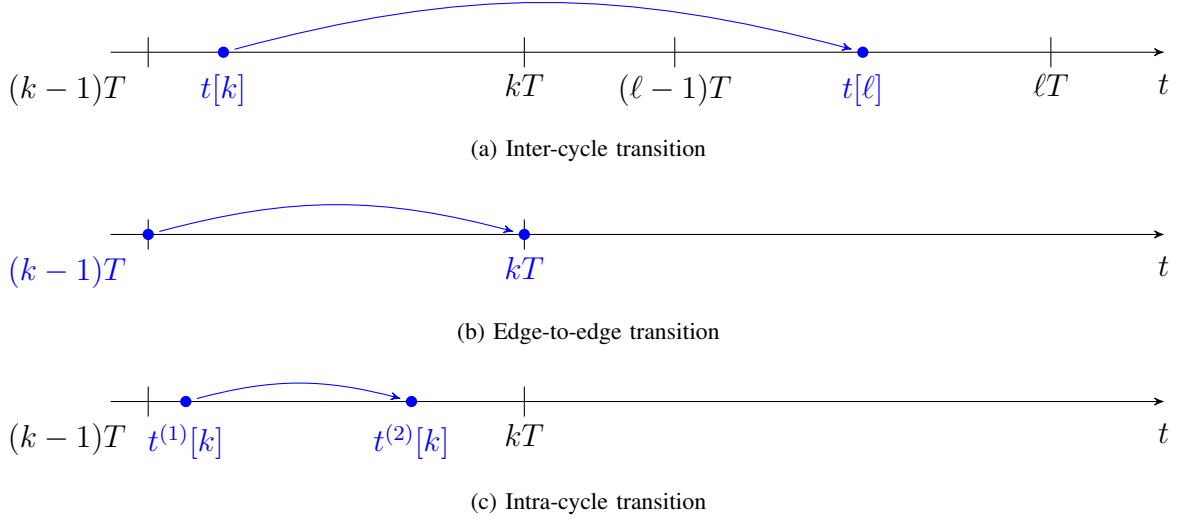

\newcommand{\mytw}{{\myt{}{1}{k}}}
\newcommand{\mytt}{{\myt{}{2}{k}}}
Due to observation noise and system uncertainty, the executive can only estimate the state of the system using observations it collects. As the underlying system is linear and observation and system noise white, it uses a Kalman filter to refine its estimates with streaming observations, which uses \eqref{eq:state_update} to predict the state an arbitrary time $t$ from an earlier estimate at time $s$. We look at two special cases: (1) when the 2 time instances are from the same decision cycle, and (2) when the 2 instances are the edges of a cycle. We denote these 2 time instances by $\mytw$ and $\mytt$, $\mytw\leq \mytt$. In the second case, the 2 instances are equal to $kT$ and $(k+1)T$.
For \textit{inter-cycle transitions}, the state transition equation reduces to
\begin{align}
\begin{split}
\vx\Big(\mytt\Big) &= \\
\mxphi\Big(\mytw,\mytt\Big)\vx\Big(\mytw\Big) &+ \mxlambda\Big(\mytw,\mytt,\mytt\Big)\vu[k] + \vnu\Big(\mytw,\mytt\Big).
\end{split}
\end{align}
For \textit{edge-to-edge transitions}, the state transition equation reduces to
\begin{align}
\begin{split}
\vx\Big((k+1)T\Big) &=  \\
\mxphi\Big(kT,(k+1)T\Big)\vx\Big(kT\Big) &+ \mxlambda\Big(kT,(k+1)T,(k+1)T\Big) \vu[k] + \vnu\Big(kT,(k+1)T\Big).
\end{split}
\end{align}
\figref{fig:kalman_step} visualizes these three cases.

\subsection{Kalman Filtering with Two Observers}\label{s:b2_2}


Maximizing the knowledge of the state of the system is equivalent to maximizing the precision of the state estimator, i.e. minimizing its \gls{mse}. For a linear system, the executive uses a Kalman filter to determine the optimal state estimate from streaming observations. We describe an example that evaluates different state estimators that rely on different numbers of observations. In this example, we assume that there are $N=2$ observers. These two observers could be a temperature and a pressure sensor, a GPS and a gyroscope signal, or even a pulse sensor and a blood glucose monitor \cite{yilmaz:2010}. Consider an arbitrary decision cycle $k$ in which the two observers produce an observation each at $t_1[k]$ and $t_2[k]$, such that $t_1[k] \leq t_2[k]$. Suppressing the cycle index $k$ and combining \eqref{eq:obs} and \eqref{eq:ss_output}, we express the 2 observations as
\begin{align}
\begin{split}
y_1 = \vc_1^\transpp \vx(t_1) + \ve_1^\transpp \vw(t_1) \\
y_2 = \vc_2^\transpp \vx(t_2) + \ve_2^\transpp \vw(t_2) 
\end{split}
\end{align} 
where $\vc_1^\transpp$ and $\vc_2^\transpp$ are the row vectors of the observation matrix $\mxc$, and $\ve_1$ and $\ve_2$ are the standard basis (column) vectors of the Euclidean plane. For notational simplicity, define $r_1$ and $r_2$ to be the variances of $\ve_1^\transpp \vw(t_1)$ and $\ve_2^\transpp \vw(t_2)$.

We assess the estimates of $\vx(kT)$ from different observation combinations: a) observation $1$ only, b) observation $2$ only, c) observations $1$ and $2$, and d) no observation. To derive these estimates, we use standard definitions from Kalman filtering: $s<t$ are two time instances; $\vxh(t|s)$ is the predicted (a priori) state estimate at time $t$ given observations prior to and including time $s$; $\vxh(t|t)$ is the filtered (a posteriori) state estimate given observations up to and including time $t$; $\mxp(t|s)$ is the a priori error covariance matrix, where the error is difference between the true state and the estimated state; $\mxp(t|t)$ is the predicted a posteriori covariance matrix; $\ve(t|s)$ is the variance of \textit{innovation} at time $t$, i.e. the difference between $\vy(t)$ and its estimate from observations prior to and including time $s$. Suppose that a state estimate $\vxh(t_0|t_0)$ and its error covariance matrix $\mxp(t_0|t_0)$ are available from an earlier cycle.

\paragraph{Only observation $1$ is harvested} The executive chooses to harvest only the observation $y_1$. Accordingly, the standard \gls{kf} predict and update steps at $t_1$ from prior observations are
\begin{align}
\mxp(t_1|t_0) &= \mxphi(t_0,t_1) \mxp(t_0|t_0)\mxphi^\transpp(t_0,t_1) +\mxq(t_0,t_1) \label{eq:kf_t1_start}\\
\ve(t_1|t_0) &= \vc_1^\transpp \mxp(t_1|t_0) \vc_1 + r_1 \\
\mxp(t_1|t_1) &= \mxp(t_1|t_0) - \frac{1}{\ve(t_1|t_0)}\mxp(t_1|t_0)\vc_1\vc_1^\transpp\mxp(t_1|t_0). \label{eq:kf_t1_stop}
\end{align}
To characterize how the covariance matrix of the predicted state at time $kT$ is refined by observations made at earlier times, we define the functions $h_i$, $h_{i,j}$ and $g_{i,j}: \mathbb{S}^L_+ \rightarrow \mathbb{S}^L_+$ as
\begin{align}
h_i(\mxp) &\triangleq \mxphi(t_i,kT) \,\mxp\, \mxphi^\transpp(t_i,kT) +\mxq(t_i,kT)\\
h_{i,j}(\mxp) &\triangleq \mxphi(t_i,t_j) \,\mxp\, \mxphi^\transpp(t_i,t_j) +\mxq(t_i,t_j)\\
g_{i,j}(\mxp) &\triangleq h_{i,j}(\mxp)
\left(\mxi-\frac{1}{\vc_j^\transpp h_{i,j}(\mxp)\vc_j + r_j}\vc_j\vc_j^\transpp h_{i,j}(\mxp) \right).
\end{align} 
The function $h_{i,j}$ computes the covariance matrix of the \emph{predicted} state at $t_j$ from the covariance matrix $\mxp(t_i|t_i)$ of the filtered estimate at $t_i$ from all observations up to and including $t_i$ (a posteriori to a priori). Setting $t_j$ to be the end of the current decision cycle $kT$, the function $h_{i,j}$ reduces to the function $h_i$. The function $g_{i,j}$ computes the covariance matrix of the \emph{filtered} state at $t_j$ from the covariance matrix $\mxp(t_i|t_i)$ of the filtered estimate at $t_i$ from all observations up to and including $t_i$ (a posteriori to a posteriori). We can relate the aposteriori error covariance matrices $\mxp(t_1|t_1)$ to $\mxp(t_0|t_0)$ according to \cite{shi:2010} as
\begin{align}
\mxp(t_1|t_1) &= g_{0,1}\big(\mxp(t_0|t_0)\big),
\intertext{and $\mxp(kT|t_1)$ to $\mxp(t_1|t_1)$ as}
\mxp(kT|t_1) &= h_1\big(\mxp(t_1|t_1)\big).
\end{align}
Finally, the estimation mean square error when only observation $1$ is used is computed as
\begin{align}\label{eq:mse_1}
\mse_{(1)} = \tr{\mxp(kT|t_1)} = \tr{h_1\left(g_{0,1}\left(\mxp(t_0|t_0)\right)\right)},
\end{align}
where the subscript $(1)$ denotes the 1-tuple (singleton) of the element $1$. We use tuples instead of sets as an index to maintain an order between their elements.

\paragraph{Only observation $2$ is harvested} When only observation $2$ is used is
\begin{align} \label{eq:mse_2}
\mse_{(2)} 
= \tr{\mxp(kT|t_2)} 
= \tr{h_2 \left(g_{0,2}\left(\mxp(t_0|t_0)\right)\right)}.
\end{align}

\paragraph{Both observation $1$ and $2$ harvested} In this case, there are two \gls{kf} update steps. The first update step, at $t_1$, is identical to \eqref{eq:kf_t1_start}--\eqref{eq:kf_t1_stop}. The second update step, at time $t_2$, is
\begin{align}
\mxp(t_2|t_1) &= \mxphi(t_1,t_2) \mxp(t_1|t_1)\mxphi^\transpp (t_1,t_2) +\mxq(t_1,t_2) \\
\mxs(t_2|t_1) &= \vc_2^\transpp \mxp(t_2|t_1) \vc_2 + r_2 \\
\mxp(t_2|t_2) &= \mxp(t_2|t_1) - \frac{1}{\ve(t_2|t_1)}\mxp(t_2|t_1)\vc_2\vc_2^\transpp \mxp(t_2|t_1) \\
&= g_{1,2}\big(  g_{0,1}\big(\mxp(t_0|t_0)\big)  \big).
\end{align}
Using the function composition operator, $\circ$, the estimation mean square error when both observations are used is
\begin{align}\label{eq:mse_12}
\begin{split}
\mse_{(1,2)} 
&=  \big( \trno \circ h_2 \circ g_{1,2}\big) \big(\mxp(t_1|t_1)\big) \\
&= \big( \trno \circ h_2 \circ g_{1,2} \circ g_{0,1} \big) \big(\mxp(t_0|t_0)\big).
\end{split}
\end{align}

\paragraph{No observation is harvested} Since there are no updates, the mean square error is exactly what it was at time $t_0$, i.e.
\begin{align}\label{eq:mse_empty}
\mse_{\emptyset} = \tr{h_0\big(\mxp(t_0|t_0)\big)},
\end{align}
where $\emptyset$ denotes the empty tuple.

With no constraints limiting the choice of the observation sequence, the executive chooses the one sequence out of the four with the minimum \gls{mse}. This is the \textit{optimal observation sequence}. For convenience, we assign variables to the error covariance matrices of the running aposteriori estimates from different observation sequences
\begin{align}
\begin{split}
\mxp_{\emptyset} &= \mxp(t_0|t_0) \\
\mxp_{(1)} &= g_{0,1}(\mxp_{\emptyset}) \\
\mxp_{(2)} &= g_{0,2}(\mxp_{\emptyset}) \\
\mxp_{(1,2)} &= g_{1,2}  (\mxp_{(1)}).
\end{split}
\end{align}
We refer to $\mxp_\vs$ as the \textit{running error covariance matrix} from the observation sequence $\vs$. In summary, \figref{fig:kf_example} visualizes the refinement of the initial covariance matrix $\mxp(t_0|t_0)$ with anywhere from zero to two observations. A caveat here is that \figref{fig:kf_example} shows how the \glspl{mse} of the four different estimators is computed; it does not show how the actual estimates are computed.

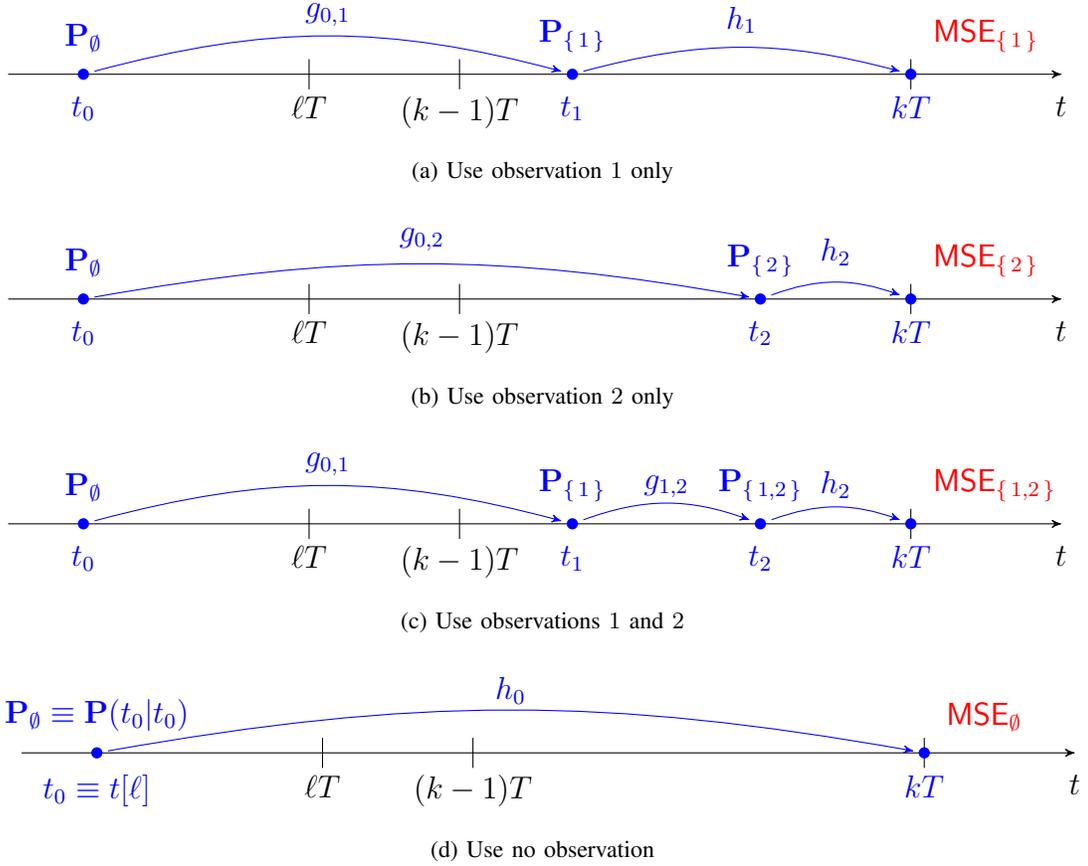
\begin{figure*}
\centering
\subfloat[Use observation $1$ only]{
\begin{tikzpicture}[>=stealth']
\def\T{6}
\def\L{14}
\draw[->] (0,0) -- (\L,0);
\node at (\L,0) [label=below:$t$] {};

\coordinate (l2) at (4,0);
\coordinate (k1) at ($(l2) + (2,0)$);
\coordinate (k2) at ($(k1) + (\T,0)$);
\coordinate (o)  at ($(k1) - (1,0)$);
\coordinate (t0)  at ($(l2) - (3,0)$);
\coordinate (t1)  at ($(k1) + (1.5,0)$);
\coordinate (t2)  at ($(k2) - (2,0)$);

\foreach \k in {l2, k1, k2}
\draw ($(\k)-(0,0.2)$) -- ($(\k)+(0,0.2)$);
\node at (l2) [label=below:$\ell T$] {};
\node at (k1) [label=below:$(k-1)T$] {};
\node (last) at (k2) 
[
label=below:\color{blue} $kT$,
label=above right:\color{red} $\mse_{\set{1}}$,
] {};

\foreach \n in {t0,t1,k2}
	\node at (\n)[blue,circle,fill,inner sep=1.5pt]{};
\node (0) at (t0) 
[
label=below:{\color{blue} $t_0$},
label=above:\color{blue} $\mxp_{\emptyset}$,
] {};
\node (1) at (t1) 
[
label=below:{\color{blue} $t_1$},
label=above:{\color{blue} $\mxp_{\set{1}}$}
] {};

\draw[->,blue] (0) to [out=15,in=165] (1); 
\draw[->,blue] (1) to [out=15,in=165] (last); 

\node[text=blue] (a) at ($(t0)!0.5!(t1)$) [above=0.5cm] {{$g_{0,1}$}};
\node[text=blue] (b) at ($(t1)!0.5!(k2)$) [above=0.4cm] {$h_{1}$};
\end{tikzpicture}
}

\subfloat[Use observation $2$ only]{
\begin{tikzpicture}[>=stealth']
\def\T{6}
\def\L{14}
\draw[->] (0,0) -- (\L,0);
\node at (\L,0) [label=below:$t$] {};

\coordinate (l2) at (4,0);
\coordinate (k1) at ($(l2) + (2,0)$);
\coordinate (k2) at ($(k1) + (\T,0)$);
\coordinate (o)  at ($(k1) - (1,0)$);
\coordinate (t0)  at ($(l2) - (3,0)$);
\coordinate (t1)  at ($(k1) + (1.5,0)$);
\coordinate (t2)  at ($(k2) - (2,0)$);

\foreach \k in {l2, k1, k2}
\draw ($(\k)-(0,0.2)$) -- ($(\k)+(0,0.2)$);
\node at (l2) [label=below:$\ell T$] {};
\node at (k1) [label=below:$(k-1)T$] {};
\node (last) at (k2) 
[
label=below:\color{blue} $kT$,
label=above right:\color{red} $\mse_{\set{2}}$,
] {};

\foreach \n in {t0,t2,k2}
	\node at (\n)[blue,circle,fill,inner sep=1.5pt]{};
\node (0) at (t0) 
[
label=below:{\color{blue} $t_0$},
label=above:{\color{blue} $\mxp_{\emptyset}$}
] {};
\node (2) at (t2) 
[
label=below:{\color{blue} $t_2$},
label=above:{\color{blue} $\mxp_{\set{2}}$}
] {};

\draw[->,blue] (0) to [out=10,in=170] (2); 
\draw[->,blue] (2) to [out=20,in=160] (last); 

\node[text=blue] (a) at ($(t0)!0.5!(t2)$) [above=0.5cm] {{$g_{0,2}$}};
\node[text=blue] (b) at ($(t2)!0.5!(k2)$) [above=0.3cm] {$h_{2}$};
\end{tikzpicture}
}

\subfloat[Use observations $1$ and $2$]{
\begin{tikzpicture}[>=stealth']
\def\T{6}
\def\L{14}
\draw[->] (0,0) -- (\L,0);
\node at (\L,0) [label=below:$t$] {};

\coordinate (l2) at (4,0);
\coordinate (k1) at ($(l2) + (2,0)$);
\coordinate (k2) at ($(k1) + (\T,0)$);
\coordinate (o)  at ($(k1) - (1,0)$);
\coordinate (t0)  at ($(l2) - (3,0)$);
\coordinate (t1)  at ($(k1) + (1.5,0)$);
\coordinate (t2)  at ($(k2) - (2,0)$);

\foreach \k in {l2, k1, k2}
	\draw ($(\k)-(0,0.2)$) -- ($(\k)+(0,0.2)$);
\node at (l2) [label=below:$\ell T$] {};
\node at (k1) [label=below:$(k-1)T$] {};
\node (last) at (k2) 
[
label=below:\color{blue} $kT$,
label=above right:\color{red} $\mse_{\set{1,2}}$,
] {};

\foreach \n in {t0,t1,t2,k2}
	\node at (\n)[blue,circle,fill,inner sep=1.5pt]{};
\node (0) at (t0) 
[
label=below:{\color{blue} $t_0$},
label=above:{\color{blue} $\mxp_{\emptyset}$},
] {};
\node (1) at (t1) 
[
label=below:{\color{blue} $t_1$},
label=above:{\color{blue} $\mxp_{\set{1}}$},
] {};
\node (2) at (t2) 
[
label=below:{\color{blue} $t_2$},
label=above:{\color{blue} $\mxp_{\set{1,2}}$},
] {};

\draw[->,blue] (0) to [out=15,in=165] (1); 
\draw[->,blue] (1) to [out=20,in=160] (2); 
\draw[->,blue] (2) to [out=20,in=160] (last); 

\node[text=blue] (a) at ($(t0)!0.5!(t1)$) [above=0.5cm] {{$g_{0,1}$}};
\node[text=blue] (b) at ($(t1)!0.5!(t2)$) [above=0.2cm] {$g_{1,2}$};
\node[text=blue] (c) at ($(t2)!0.5!(k2)$) [above=0.2cm] {$h_{2}$};
\end{tikzpicture}
}

\subfloat[Use no observation]{
\begin{tikzpicture}[>=stealth']
\def\T{6}
\def\L{14}
\draw[->] (0,0) -- (\L,0);
\node at (\L,0) [label=below:$t$] {};

\coordinate (l2) at (4,0);
\coordinate (k1) at ($(l2) + (2,0)$);
\coordinate (k2) at ($(k1) + (\T,0)$);
\coordinate (o)  at ($(k1) - (1,0)$);
\coordinate (t0)  at ($(l2) - (3,0)$);
\coordinate (t1)  at ($(k1) + (1.5,0)$);
\coordinate (t2)  at ($(k2) - (2,0)$);

\foreach \k in {l2, k1, k2}
\draw ($(\k)-(0,0.2)$) -- ($(\k)+(0,0.2)$);
\node at (l2) [label=below:$\ell T$] {};
\node at (k1) [label=below:$(k-1)T$] {};
\node (last) at (k2) 
[
label=below:\color{blue}$kT$, 
label=above right:\color{red}$\mse_{\emptyset}$
] {};

\foreach \n in {t0,k2}
\node at (\n)[blue,circle,fill,inner sep=1.5pt]{};
\node (0) at (t0) 
[
label=below:{\color{blue} $t_0\equiv t[\ell]$}, 
label=above:{\color{blue} $\mxp_{\emptyset}\equiv \mxp(t_0|t_0)$}
] {};

\draw[->,blue] (0) to [out=10,in=170] (last); 

\node[text=blue] (a) at ($(t0)!0.5!(k2)$) [above=0.5cm] {{$h_0$}};
\end{tikzpicture}
}

\caption{Evaluating the estimate of $\vx(t)$ using at most 2 harvested observations. 
In (a), the estimate at $t_0\equiv t[\ell]$ is filtered with the observation at $t_1$ to estimate the state at $t_1$ which is then used to predict the state at $kT$. To compute the \gls{mse} of the estimator using an observation at $t_1$, $\mxp_{\emptyset}$ is used to compute $\mxp_{\set{1}}$, and then $\mxp_{\set{1}}$ is used to compute $\mse_{\set{1}}$ 
In (b), the observation at $t_2$ is used instead of that at $t_1$. 
In (c), the estimate at $t_0$ is filtered with the observation at $t_1$ to estimate the state at $t_1$; then, the estimate at $t_1$ is filtered with the observation at $t_1$ to estimate the state at $t_2$. Finally, the estimate at $t_2$ is used to predict the state at $kT$. To compute the \gls{mse} of this estimator, $\mxp_{\emptyset}$ is used to compute $\mxp_{\set{1}}$, and $\mxp_{\set{1}}$ is used to compute $\mxp_{\set{1,2}}$. Finally, $\mxp_{\set{1,2}}$ is used to compute $\mse_{\set{1,2}}$
In (d), as the executive harvests no observations, so it uses $\mxp_{\emptyset}$, the aposteriori covariance matrix of the state estimate formed at $t_0$ to compute the \gls{mse} of the predicted state at $kT$ from that at $t_0$.}
\label{fig:kf_example}
\end{figure*}

We refer to the previous example to highlight two key points in deriving the algorithm that finds the optimal observation sequence (covered in \secref{s:b4}).
First, the executive evaluates all possible observation sequences, $(1)$, $(2)$, $(1,2)$ and $\emptyset$ \emph{probabilistically}, i.e. without knowing the value of a single observation. Indeed, the Bayesian framework allows for evaluating different estimators without knowing the content of the very observations that go into these estimators, only their timestamps. Starting with an error covariance $\mxp_\emptyset$ and given the set of observation timestamps, a number of deterministic functions are applied to evaluate the \gls{mse} for all observation sequences. The sequence that yields the smallest \gls{mse} is deemed optimal.
Second, keeping track of the running covariance matrix for an observation sequence provides a computational speedup when computing the \gls{mse} for a longer sequence. The \gls{mse} for an observation sequence, e.g. $\set{1,2,3}$, cannot be computed from that of its longest prefix, i.e. $\set{1,2}$, but rather from the running covariance matrix of that prefix.
Consider for example $\mse_{\set{1}}$ and $\mse_{\set{1,2}}$ given as
\begin{align}
\mse_{(1)} &=  \big( \trno \circ h_1 \circ g_{0,1} \big)  \big(\mxp_\emptyset\big), \\
\mse_{(1,2)} &= \big( \trno \circ h_2 \circ g_{1,2} \circ g_{0,1} \big) \big(\mxp_\emptyset\big).
\end{align}
As the functions $h_2$, $g_{1,2}$ and $g_{0,1}$ are not commutative (also, $h_1$ and $g_{0,1}$) and the trace operation lossy, there is no clear relationship between the two \glspl{mse}. There is, however, a relationship between $\mse_{\set{1,2}}$, $\mxp_{\set{1,2}}$, and $\mxp_{\set{1}}$:
\begin{align}
\mxp_{\set{1,2}} &= g_{1,2} \big(\mxp_{\set{1}}\big),\\
\mse_{(1,2)} &= \big( \trno \circ h_2 \circ g_{1,2} \big) \big(\mxp_{\set{1}}\big).
\end{align}
By keeping track of the running covariance matrix for every observation sequence, the \gls{mse} for any desired sequence can be quickly computed from the running covariance matrix for its longest prefix.

\section{Optimization Problem}\label{s:b3}
In this section, we define the objective function through which the executive selects the optimal observation sequence and derive the expression for the constraint that limits this selection.

\subsection{Channel and Medium Access Models}
We make a number of assumptions on the channel, medium access, and scheduling models for the purpose of analytical tractability. We assume that all observer-executive and executive-agent channels are block fading wireless channels with coherence times equal to the decision cycle period. Since the channel state determines the transmission rate and airtime (transmission time), these two quantities change every cycle as well. We assume that the executive has perfect knowledge of all the channels in every cycle and of the sizes of the observation and action payloads. In practice, this can be accomplished by prepending the observation and action packets by pilot sequences for the purpose of channel estimation. 

At the start of every decision cycle, we assume the executive solves an optimization problem for the optimal observation sequence. Without loss of generality, we look at an arbitrary decision cycle with $L$ representative observations available for harvesting, one per observer. As observation timestamps are known offline (see \secref{s:b1}), we define $\set{t_\ell}_{\ell=1}^L$ to be the timestamps of the representative observations of the $L$ observers. We suppress the dependence of these timestamps, and other variables to come, on the decision cycle index $k$. We define $\set{O_\ell}_{\ell=1}^L$ and $\set{A_m}_{m=1}^M$ to be the airtimes for transmitting the observations and actions. We assume that the observations are transmitted on a \gls{fcfs} basis with no preemption. We also assume that the actions are transmitted once the observations are received. We do not specify a service policy for transmitting actions, but we assume that there are no breaks between these transmissions. To distinguish the order of observations from their identities while maintaining notational clarity, we order the set of observation timestamps $\set{t_\ell}_{\ell=1}^L$ in any given cycle into the set $\set{o_\ell}_{\ell=1}^L$ so that
\begin{align}
\begin{split}
o_1 &= \min\set{t_1,\dots,t_{N}} \\
o_\ell &= \min\set{t_1,\dots,t_{N}}\setminus\set{o_1,\dots,o_{\ell-1}}.
\end{split}
\end{align}
Similar to \secref{s:b2}, we define $o_0=t_0$ to be the timestamp of the latest observation from a prior decision cycle.

\subsection{Constrained Selection}
Every decision cycle, the executive solves an optimization problem whose outcome is an observation sequence that maximizes its knowledge of the current state of the system. What constrains the executive's choice of such a sequence is the limited time window in which it has to execute all of its functions: receiving observation packets, estimating the current state of the system, determining corrective actions, and transmitting action packets. Therefore, the executive can only harvest the observations that leave ample time to perform the remaining functions. In \secref{s:b2}, we looked at an example that derives the \gls{mse} of different state estimates from all possible observation sequences. Now, we derive the expression for the constraint that limits the choice of the optimal observation sequence. We first derive the constraint for when there are $L=3$ representative observations. We later rederive the constraint when $L$ is arbitrary.

When 3 observations are available, there are a total of $2^3=$ 8 possible choices with running error covariance matrices
\begin{alignat*}{3}
&\mxp_{\emptyset}=\mxp(o_0\,|\,o_0), \\
&\mxp_{(1)} =   g_{0,1} \; (\mxp_\emptyset),
&&\mxp_{(2)} =   g_{0,2} \; (\mxp_\emptyset),
&&\mxp_{(3)} =   g_{0,3} \; (\mxp_\emptyset), \\
&\mxp_{(1,2)} =   g_{0,1} \circ g_{1,2} \; (\mxp_\emptyset),
&&\mxp_{(1,3)} =   g_{0,1} \circ g_{1,3} \; (\mxp_\emptyset),
&&\mxp_{(2,3)} =   g_{0,2} \circ g_{2,3} \; (\mxp_\emptyset), \\
&\mxp_{(1,2,3)} =   g_{0,1} \circ g_{1,2} \circ g_{2,3}\; (\mxp_\emptyset).
&&\phantom{\mxp_{(1,2,3)} =   g_{0,1} \circ g_{1,2} \circ g_{2,3}\; (\mxp_\emptyset),}
\end{alignat*}

We consider two constraints jointly: a \textit{latency} constraint and a \textit{dependency} constraint. According to the latency constraint, the observations and actions shall be exchanged by the end of the decision cycle. According to the dependency constraint, the actions shall be dispatched only after the selected observations are received, reflecting the assumption that every action is calculated as a function of all observation. In future work, we will relax this assumption by letting every action be a function of as little as one observation. Accordingly, an action can be dispatched over the shared wireless channel as soon as its influencing observations are harvested. If observations are sufficiently spaced apart so that their airtimes are nonoverlapping, then dispatching actions would fill in these gaps and increase channel utilization.

Any observation sequence selected by the executive should satisfy the latency and dependency constraints. Suppose that the sequence $\vs=(1,2,3)$ has already been selected by the executive for harvesting. We define $d_\ell$, $1\leq \ell\leq 3$, to be the timestamp when the observation originating at $o_\ell$ is completely transmitted. Since observations are transmitted in order of availability, the timestamps $d_1$, $d_2$, and $d_3$ are related as
\begin{align}
\begin{split}
d_1 &= o_1 + O_1 \\
d_2 &= \max\set{ o_2 + O_2,\;\; d_1 + O_2 } \\
d_3 &= \max\set{ o_3 + O_3,\;\; d_2 + O_3 }.
\end{split}
\end{align}
In the best-case scenario, $d_\ell=o_\ell+O_\ell$ if $o_\ell$ is available when no prior observation is being harvested. In the worst-case scenario, $d_\ell=o_1+O_1+\dots+O_\ell$ if $o_\ell$ is produced while $o_1$ is being harvested; $o_\ell$ has to wait for all prior observations to be harvested. In all cases, $d_{3}$ is the time when \emph{all} observations are harvested and ready to be processed. We refer to $d_3$ as the \textit{end-of-harvest time}. \figref{fig:blockages} shows time diagrams depicting a range of scenarios for harvesting $(1,2,3)$. Every observation is produced at the same time across all 3 diagrams but has a varying airtime. In \figref{fig:blockages_a}, observations are sufficiently spaced apart and their airtimes sufficiently short so that every observation is harvested as soon as it is produced. All 3 observations are harvested at $d_{3}=o_3+O_3$. In \figref{fig:blockages_b}, $o_2$ blocks $o_3$, so $d_3=o_2+O_2+O_3$. In \figref{fig:blockages_c}, there is a domino effect: $o_1$ blocks $o_2$, and $o_2$ blocks $o_3$; consequently, $o_2$ and $o_3$ have to wait and $d_3=o_1+O_1+O_2+O_3$. The values of $\set{o_n}$ and $\set{O_n}$ are known to the executive prior to selection, so it can compute an expression for $d_2$ and $d_3$ as follows
\begin{align}
d_2 &= 
\begin{cases}
o_2+O_2     & \mbox{if $o_2 > d_1$,} \\
o_1+O_1+O_2 \phantom{+O_3} & \mbox{if $o_2 \leq d_1$.}
\end{cases}
\\
d_3 &= 
\begin{cases}
o_3+O_3         & \mbox{if $o_3 > d_2$,} \\
o_2+O_2+O_2     & \mbox{if $o_3 \leq d_2$ and $o_2 > d_1$,} \\
o_1+O_1+O_2+O_3 & \mbox{if $o_3 \leq d_2$ and $o_2 \leq d_1$.}
\end{cases}
\end{align}
The variables $d_2$ and $d_3$ can be expressed in a form that is agnostic of the different scenarios of \figref{fig:blockages} as
\begin{align}\label{eq:d3}
\begin{split}
d_2 &= \max\set{o_2+O_2,\;o_1+O_1+O_2} \\
d_3 &= \max\set{o_3+O_3,\;o_2+O_2+O_3,\;o_1+O_1+O_2+O_3}.
\end{split}
\end{align}
We have just determined the time it takes to harvest the observation sequence. The remaining time in the decision cycle is dedicated to dispatching the actions. As the actions are dispatched back-to-back, the constraint finally is
\begin{align}\label{eq:constraint}
d_3 + \sum_{m=1}^{M} A_m < T.
\end{align}

The end-of-harvest time $d_3$ assumes that $(1,2,3)$ has been already selected for harvesting, so the resulting constraint is tailored to that particular observation sequence. This defeats the purpose for having a constraint which is steering the selection process. While $o_\ell$ is known to the executive prior to running the selection process, $d_\ell$ is determined one observation at a time through the selection process. We thus express $d_{|\vs|}$ for an arbitrary observation sequence $\vs$ through the mapping $\vs \mapsto d_{|\vs|}$ as follows
\begin{align}
\begin{split}
\emptyset \;&\; \mapsto 0 \\
(1) \;&\; \mapsto \max\set{0,\;o_1+O_1} \\
(2) \;&\; \mapsto \max\set{0,\;o_2+O_2} \\
(3) \;&\; \mapsto \max\set{0,\;o_3+O_3} \\
(1,2) \;&\; \mapsto \max\set{0,\;o_2+O_2,\;o_1+O_1+O_2} \\
(1,3) \;&\; \mapsto \max\set{0,\;o_3+O_3,\;o_1+O_1+O_3} \\
(2,3) \;&\; \mapsto \max\set{0,\;o_3+O_3,\;o_2+O_2+O_3} \\
(1,2,3) \;&\; \mapsto \max\set{0,\;o_3+O_3,\;o_2+O_2+O_3,\;o_1+O_1+O_2+O_3}. \\
\end{split}
\end{align}
With the end-of-harvest time redefined, an arbitrary sequence $\vs$ must satisfy the constraint
\begin{align}\label{eq:constraint_2}
d_{|\vs|} + \sum_{m=1}^{M} A_m < T.
\end{align}


%
%

\newcommand\T{15}

\begin{figure*}[t]
\centering
\captionsetup[subfigure]{oneside,margin={2cm,0cm}}
\subfloat[$o_3$ is unblocked]{\label{fig:blockages_a}
\begin{scaletikzpicturetowidth}{\textwidth}
\begin{tikzpicture}[scale=\tikzscale,xscale=1,transform shape]
\draw [-latex'](-0.5,0) -- (0,0)coordinate(o) -- (\T+2,0)coordinate(tmax) node[above]{$t$};
\draw[thick] (0,0.2) -- (0,-0.2) node[below]{$0$};
\draw[thick] (\T,0.2) -- (\T,-0.2) node[below]{$T$};

\coordinate(o1) at (1,0);
\node at (o1|-,-0.2) [above=0.4, left]{$o_1$};
\coordinate(d1) at (5,0);
\node at (d1|-,-0.2) [below left]{$d_1$};

\draw[fill=gray!20!white] (o1) -- (o1|-,0.8) -- (d1|-,0.8) -- (d1) -- cycle;

\draw[-latex,thick] (o1.north west|-,1.5) --(o1);
\draw[-latex,thick] (d1) -- (d1.north west|-,-0.7);

\coordinate(mid1) at ($(o1)!0.5!(d1)$);
\node at (mid1|-,1.1){$O_1$};

\coordinate(o2) at (6,0);
\node at (o2|-,-0.2) [above=0.4, left]{$o_2$};
\coordinate(d2) at (7,0);
\node at (d2|-,-0.2) [below left]{$d_2$};

\draw[fill=gray!20!white] (o2) -- (o2|-,0.8) -- (d2|-,0.8) -- (d2) -- cycle;
\coordinate(mid2) at ($(o2)!0.5!(d2)$);
\node at (mid2|-,1.1){$O_2$};

\draw[-latex,thick] (o2.north west|-,1.5) --(o2);
\draw[-latex,thick] (d2) -- (d2.north west|-,-0.7);

\coordinate(o3) at (7.5,0);
\node at (o3|-,-0.2) [above=0.4, left]{$o_3$};
\coordinate(d3) at (10,0);
\node at (d3|-,-0.2) [below left]{$d_3$};

\draw[fill=gray!20!white] (o3) -- (o3|-,0.8) -- (d3|-,0.8) -- (d3) -- cycle;
\coordinate(mid3) at ($(o3)!0.5!(d3)$);
\node at (mid3|-,1.1){$O_3$};

\draw[-latex,thick] (o3.north west|-,1.5) --(o3);
\draw[-latex,thick] (d3) -- (d3.north west|-,-0.7);
\end{tikzpicture}	
\end{scaletikzpicturetowidth}	
}

\subfloat[$o_2$ blocks $o_3$]{\label{fig:blockages_b}
\begin{scaletikzpicturetowidth}{\textwidth}
\begin{tikzpicture}[scale=\tikzscale,xscale=1,transform shape]
\draw [-latex'](-0.5,0) -- (0,0)coordinate(o) -- (\T+2,0)coordinate(tmax) node[above]{$t$};
\draw[thick] (0,0.2) -- (0,-0.2) node[below]{$0$};
\draw[thick] (\T,0.2) -- (\T,-0.2) node[below]{$T$};

\coordinate(o1) at (1,0);
\node at (o1|-,-0.2) [above=0.4, left]{$o_1$};
\coordinate(d1) at (5,0);
\node at (d1|-,-0.2) [below left]{$d_1$};

\draw[fill=gray!20!white] (o1) -- (o1|-,0.8) -- (d1|-,0.8) -- (d1) -- cycle;
\coordinate(mid1) at ($(o1)!0.5!(d1)$);

\draw[-latex,thick] (o1.north west|-,1.5) --(o1);
\draw[-latex,thick] (d1) -- (d1.north west|-,-0.7);

\coordinate(o2) at (6,0);
\node at (o2|-,-0.2) [above=0.4, left]{$o_2$};
\coordinate(d2) at (8,0);
\node at (d2|-,-0.2) [below left]{$d_2$};

\draw[fill=gray!20!white] (o2) -- (o2|-,0.8) -- (d2|-,0.8) -- (d2) -- cycle;
\coordinate(mid2) at ($(o2)!0.3!(d2)$);

\draw[-latex,thick] (o2.north west|-,1.5) --(o2);
\draw[-latex,thick] (d2) -- (d2.north west|-,-0.7);

\coordinate(o3) at (7.5,0);
\node at (o3|-,-0.2) [above=0.4, left]{$o_3$};
\coordinate(d3) at (10.5,0);
\node at (d3|-,-0.2) [below left]{$d_3$};

\draw[fill=gray!20!white] (d2) -- (d2|-,0.8) -- (d3|-,0.8) -- (d3) -- cycle;
\coordinate(mid3) at ($(d2)!0.5!(d3)$);

\draw[-latex,thick] (o3.north west|-,1.5) --(o3);
\draw[-latex,thick] (d3) -- (d3.north west|-,-0.7);
\end{tikzpicture}	
\end{scaletikzpicturetowidth}
}

\subfloat[$o_2$ blocks $o_3$ and $o_2$ blocks $o_3$]{\label{fig:blockages_c}
\begin{scaletikzpicturetowidth}{\textwidth}
\begin{tikzpicture}[scale=\tikzscale,xscale=1,transform shape]
\draw [-latex'](-0.5,0) -- (0,0)coordinate(o) -- (\T+2,0)coordinate(tmax) node[above]{$t$};
\draw[thick] (0,0.2) -- (0,-0.2) node[below]{$0$};
\draw[thick] (\T,0.2) -- (\T,-0.2) node[below]{$T$};

\coordinate(o1) at (1,0);
\node at (o1|-,-0.2) [above=0.4, left]{$o_1$};
\coordinate(d1) at (6.5,0);
\node at (d1|-,-0.2) [below left]{$d_1$};

\draw[fill=gray!20!white] (o1) -- (o1|-,0.8) -- (d1|-,0.8) -- (d1) -- cycle;
\coordinate(mid1) at ($(o1)!0.5!(d1)$);

\draw[-latex,thick] (o1.north west|-,1.5) --(o1);
\draw[-latex,thick] (d1) -- (d1.north west|-,-0.7);

\coordinate(o2) at (6,0);
\node at (o2|-,-0.2) [above=0.4, left]{$o_2$};
\coordinate(d2) at (8.5,0);
\node at (d2|-,-0.2) [below left]{$d_2$};

\draw[fill=gray!20!white] (d1) -- (d1|-,0.8) -- (d2|-,0.8) -- (d2) -- cycle;
\coordinate(mid2) at ($(d1)!0.5!(d2)$);

\draw[-latex,thick] (o2.north west|-,1.5) --(o2);
\draw[-latex,thick] (d2) -- (d2.north west|-,-0.7);

\coordinate(o3) at (7.5,0);
\node at (o3|-,-0.2) [above=0.4, left]{$o_3$};
\coordinate(d3) at (11,0);
\node at (d3|-,-0.2) [below left]{$d_3$};

\draw[fill=gray!20!white] (d2) -- (d2|-,0.8) -- (d3|-,0.8) -- (d3) -- cycle;
\coordinate(mid3) at ($(o3)!0.5!(d3)$);

\draw[-latex,thick] (o3.north west|-,1.5) --(o3);
\draw[-latex,thick] (d3) -- (d3.north west|-,-0.7);
\end{tikzpicture}
\end{scaletikzpicturetowidth}
}
\caption{Three different scenarios for three observations. Observation $i$ is produced by the observer at $o_i$ and received by the executive at $d_i$; it requires $O_i$ time units to be transmitted. A downward-pointing indicates the time an observation is produced. An upward-pointing arrow indicates the time an observation is received by the executive. Observations are transmitted on a \gls{fcfs} basis, and transmissions are scheduled upon observation availability as long as no other transmission is in progress. In \figref{fig:blockages_a}, every observation is produced and fully transmitted before the next observation is produced. In \figref{fig:blockages_b}, $o_2$ completes transmission after $o_3$ is produced. In \figref{fig:blockages_c}, $o_1$ completes transmission after $o_2$ is produced, and $o_2$ completes transmission after $o_3$ is produced. }\label{fig:blockages}
\vspace{-20pt}
\end{figure*}
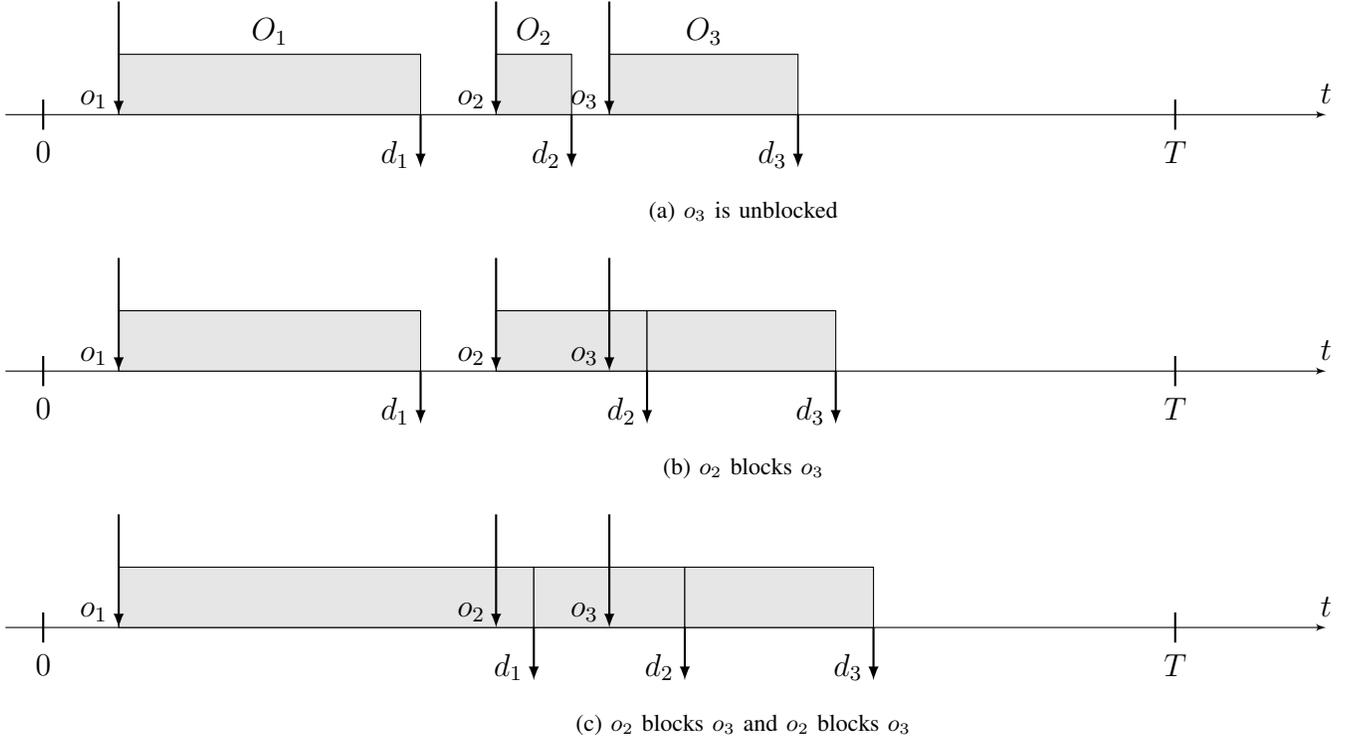

\let\T\undefined

\subsection{Objective Function and Constraint}
We generalize both the objective function, which is the estimation \gls{mse}, and the constraint to an arbitrary number of representative observations $L$.

Given that the running error covariance matrix of the latest estimate $\vxh(o_0|o_0)$ is $\mxp_\emptyset$, the executive's objective is to select the observation sequence $\vs\in\powerset{1,2,\dots,L}$ that minimizes the estimation \gls{mse}
\begin{align}\label{eq:objective_function}
\mse_{\vs}(\mxp_\emptyset) =  
\trno  
\circ h_{s_{-1}}
\circ  g_{s_{-2},s_{-1}} 
\circ \cdots
\circ g_{s_1,s_2}
\circ g_{0,s_1}
\;\; (\mxp_\emptyset).
\end{align}
The end-of-harvest time $d_{\vs}$ of an arbitrary observation sequence $\vs$ should satisfy
\begin{align}\label{eq:constraint_arbitrary}
d_{|\vs|} + \sum_{m=1}^M A_m &< T, 
\intertext{which is equivalent to}
\max_{s_1\leq i\leq s_{|\vs|}} \left( o_i + \sum_{j=i}^{|\vs|} O_j \right) &< T-\sum_{m=1}^M A_m.
\end{align}
We define the shorthand variable $B=T - \sum_{m=1}^M A_m$ to be the \textit{harvesting budget}. The sequence $\vs$ is \textit{schedulable} only if $d_{|\vs|}<B$, and \textit{non-schedulable} otherwise. As the executive knows $\set{o_\ell}$ offline and $\set{O_\ell}$ at the start of every decision cycle, it can verify whether an observation sequence satisfies the constraint.

{
\color{black}
In short, OSP is the following constrained minimization problem:
\tcbset{highlight math style={boxsep=5mm,colback=blue!1!red!1!white}}
\begin{tcolorbox}[ams align]
&\vs^* = 
\argmin_{\vs\in\powerset{1,2,\dots,L}}
\mse_{\vs}(\mxp_\emptyset) =  
\argmin_{\vs\in\powerset{1,2,\dots,L}}
\trno  
\circ h_{s_{-1}}
\circ  g_{s_{-2},s_{-1}} 
\circ \cdots
\circ g_{s_1,s_2}
\circ g_{0,s_1}
\;\; (\mxp_\emptyset)
\intertext{such that}
&\phantom{AAAAAAAAAAAAAA}\max_{s_1\leq i\leq s_{|\vs|}} \left( o_i + \sum_{j=i}^{|\vs|} O_j \right) < B.
\end{tcolorbox}
}

\section{Branch-and-Bound Algorithm}\label{s:b4}
\newcommand{\problem}{\mathcal{P}}
In this section, we propose a \gls{bnb} algorithm to optimally but efficiently solve  OSP. A \gls{bnb} algorithm is a brute-force algorithm that systematically iterates over the search space and eliminates invalid solutions. For our problem, the search space is the set of observations sequences $\powerset{1,2,\dots,L}$. An invalid solution is a \textit{non-schedulable} observation sequence $\vs$: a sequence of observations that cannot be transmitted by the hard delivery deadline. Every decision cycle, the executive runs this algorithm to determine the sequence of observations that maximizes its belief about the state of the system.

We represent the solution space as a \textit{subset forest} of trees as shown in \figref{fig:subset_forest}. Tree nodes are tagged with an observation $o_i$, $i=1,\dots,L$. Every node maps to a unique observations sequence: This sequence is determined by following the path from the root of the tree down to that node. For example, the node marked with ${\color{ForestGreen}o_5}$ in \figref{fig:subset_forest} maps to the sequence $(1,3,4,5)$, a shorthand for $(o_1,o_3,o_4,o_5)$. As every internal node (non-leaf node) has child nodes, every sequence of observations is the \textit{prefix} of another, longer sequence. If the shorter of these two sequences is deemed non-schedulable, then the longer of the two is automatically deemed non-schedulable. Accordingly, every tree in the forest can be pruned and the search space reduced by discarding sequences prefixed by non-schedulable sequences. We now describe the algorithm. 

\paragraph{Order of sequence evaluation} To eliminate potential non-schedulable observation sequences, our algorithm tests sequences of every tree from top to bottom. Additionally, we choose to traverse the trees in the order shown in \figref{fig:subset_forest}, i.e. from left to right. 

\paragraph{Choosing the first observation} The algorithm creates a sequence containing only $o_1$, the earliest observation, if it deems the sequence $(1)$ schedulable; otherwise, it determines that the first tree contains no schedulable sequence and thus no solution. It moves to subsequent trees and repeats the process. If $(1)$ is indeed schedulable, it computes the running covariance matrix $\mxp_{(1)}$ followed by $\mse_{(1)}$.

\paragraph{Choosing the second observation} The algorithm then tries to append the earliest observation $o_{k_2}$ to the existing sequence $(k_1)$ as long as $(k_1,k_2)$ is schedulable. Once the algorithm finds such an $o_{k_2}$, it computes the running covariance matrix $\mxp_{(k_1,k_2)}$ from $\mxp_{(k_1)}$ and $\mse_{(k_1,k_2)}$ from $\mxp_{(k_1,k_2)}$. 

\paragraph{Appending subsequent observations} The algorithm continues appending subsequent observations to a running sequence as long as the constraint in \eqref{eq:constraint_arbitrary} remains satisfied. At every step, it computes both the running covariance matrix and the \gls{mse}. It compares this \gls{mse} to the least encountered \gls{mse} so far, updating the best encountered sequence if the former \gls{mse} is less than the latter.

\paragraph{Pruning the search space} If adding an observation $o_i$ to a sequence $\vs$ violates the constraint, then appending any sequence of observations that starts with $o_i$ will also violate constraint. Accordingly, the algorithm ignores sequences that start with $(\vs, i)$.  Note that appending $o_j$, $j>i$, may not violate the constraint if $O_j < O_i$. In that case, the algorithm continues adding observations past $o_j$ until the constraint is violated. 

\paragraph{Termination} Once the algorithm has examined all sequences in a tree, it moves on to sequences in the next tree and repeats the same procedure. Once the algorithm has examined all all tree, it selects the observation sequence with the least \gls{mse}. 

{\color{black}
\paragraph{Complexity} The algorithm has to explore the entire search space to decide on the optimal observation sequence. Accordingly, the algorithm has to try all $2^L$ possibilities, including the null sequence, so its complexity is exponential in the number of observations, i.e. $O(2^L)$. The complexity is also exponential in the number of observers, i.e. $O(2^N)$, because every observer is limited to a single representative observation per decision cycle.
}


\newcommand\T{15}

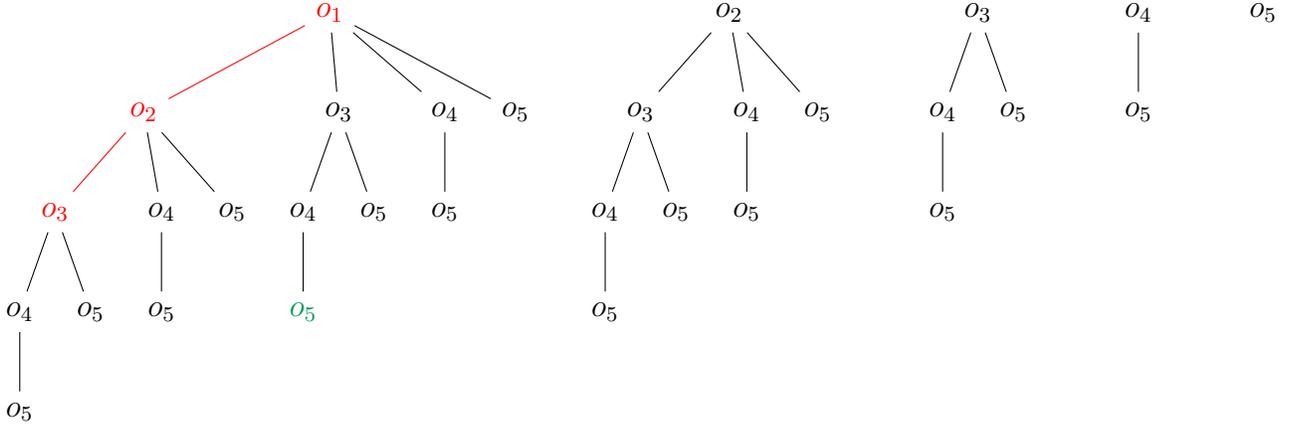
\begin{figure*}
\centering
\newcommand{\redo}[1]{ {\color{red}$o_{#1}$} }
\newcommand{\greeno}[1]{ {\color{ForestGreen}$o_{#1}$} }
\begin{forest}
[,phantom, s sep=1cm
[\redo1 [\redo2, edge={red} [\redo3 , edge={red} [$o_4$[$o_5$]][$o_5$]][$o_4$[$o_5$]][$o_5$]] [$o_3$[$o_4$[\greeno5]][$o_5$]] [$o_4$[$o_5$]] [$o_5$]]
[$o_2$ [$o_3$[$o_4$[$o_5$]][$o_5$]] [$o_4$[$o_5$]] [$o_5$]]
[$o_3$ [$o_4$[$o_5$]] [$o_5$]]
[$o_4$[$o_5$]]
[$o_5$]
]
\end{forest}

\caption{The subset forest for $L=5$ representative observations has a total of 32 observation sequences. 
The algorithm finds the optimal observation sequence by searching the forest for a schedulable sequence with the least \gls{mse}.
Every node maps to a unique observation sequence determined by following the path from the root of the tree down to that node. For example, the node marked with ${\color{ForestGreen}o_5}$ maps to the sequence $(1,3,4,5)$, a shorthand for $(o_1,o_3,o_4,o_5)$.
Extending an already nonschedulable sequence does not make for a schedulable one. For example, suppose that $(1,2,3)$, marked in red, is invalid, then so are $(1,2,3,4)$, $(1,2,3,5)$, and $(1,2,3,4,5)$. Other paths, however, like $(1,2,4)$ and $(1,2,4,5)$ are not necessarily invalid and may yield a lower \gls{mse}.}
\label{fig:subset_forest}
\end{figure*}

\let\T\undefined


\paragraph{Pseudocode} We formalize the previous description into the recursive algorithm implemented by the function \textsc{Search}($j$, $d$, $X$, $\vs$). The parameters of the algorithm are the next candidate observation index $j$, the running end-of-harvest time $d$, the running aposteriori covariance matrix $X$ of the latest observation, and the running observation sequence. To run through all sequences starting at with any observation, a function \textsc{SearchBF}($\mxp_\emptyset$) successively calls the function \textsc{Search}($j$, $0$, $\mxp_\emptyset$, $\emptyset$) is called for all $1\leq j\leq L$, and select the sequence with the least \gls{mse}. \textsc{SearchBF} has a straight forward implementation which we skip to save space.
\begin{algorithmic}
\Function{Search}{$j$, $d$, $\mathbf{X}$, $\vs$}\label{alg:bf}
	\If{$i=1$ \bf{or} $o_i>d$}
		\State $d\gets o_j + O_j$
	\Else
		\State $d\gets d+O_j$
	\EndIf
	
	\If{$d>B$}
		\State \Return $(\mathbf{X},\vs)$
	\EndIf
	
	\State $i\gets \vs_{-1}$
	\State $\mathbf{X}^*\gets \mathbf{X}$
	\State $\mathbf{P}^*\gets h_i(\mathbf{X})$
	\State $\vs^*\gets \vs$
	
	\For{$j+1\leq k\leq L$}
		\State $(\mathbf{X},\vs)=$ \Call{Search}{$k$, $d$, $g_{i,j}(\mathbf{X})$, $(\vs,j)$}
		\State $l\gets \vs_{-1}$
		\If{$\tr{h_k(\mathbf{X})} \leq \tr{\mathbf{P}^*}$}
			\State $\mathbf{X}^*\gets \mathbf{X}$
			\State $\mathbf{P}^*\gets h_k(\mathbf{X})$
			\State $\vs^*\gets \vs$	
		\EndIf
	\EndFor
	\State \Return $(\mathbf{X},\vs)$		
\EndFunction
\end{algorithmic}

{\color{black}
While the proposed \gls{bnb} algorithm is optimal, it has exponential complexity and does not scale well with the number of observers. Therefore, an approximation algorithm with reasonable time complexity is warranted. We devise a greedy algorithm to solve the problem, albeit suboptimally. The greedy algorithm lines up the observations in chronological order, $o_1,\dots,o_L$, and attempts to schedule them one after the other. The algorithm first determines the first observations $o_i$ such that the sequence $(i)$ is schedulable. Next, it finds the first observation $o_j$ it encounters such that the new sequence $(i,j)$ is schedulable. The following examples shows why this algorithm is suboptimal. Suppose that the two observations that come in time after $o_i$ are $o_j$ and then $o_k$. The algorithm evaluates the sequence $(i,j)$, as long as its schedule, before evaluating the sequence $(i,k)$, although the latter might yield a smaller \gls{mse}. 

Our greedy algorithm resembles first-come-first-serve (FCFS) scheduling, with the added subtlety that observations that are known not to make the deadline are not scheduled in the first place. In \secref{s:num}, we use the greedy algorithm as a baseline for evaluating the \gls{bnb} algorithm. 
}

\section{Numerical Example}\label{s:num}
In this section, we solve OSP for the \gls{lti} system determined by the matrices
\begin{align}
\mxa =
\begin{bmatrix}
-10.0	& 1.0	& 0 \\
-0.02	& -2.0	& 156.3 \\
0 	& 0 	& -1000
\end{bmatrix},
\;\;
\mxb = 
\begin{bmatrix}
0 \\
0 \\
64 
\end{bmatrix}.
\end{align}
The decision period is $T=0.01$. For the system uncertainty matrix, we use $\mxq = 10^{-2} \mxi$. We use two observation matrices
\begin{align}
\mxc_1 &=
\begin{bmatrix}
1 & 1 & 0 & 0 & 0 & 0 \\
0 & 0 & 1 & 1 & 0 & 0 \\
0 & 0 & 0 & 0 & 1 & 1
\end{bmatrix}^\mathsf{T}
\intertext{and}
\mxc_2 &=
\begin{bmatrix}
-0.684 & -0.684 & 0.504 & 0.504 & 2.180 & 2.180 \\
0.763 & 0.763& 0.765 & 0.765 & -0.554 & -0.554 \\
0.144 & 0.144 & 0.532 & 0.532 &  -0.632 &  -0.632
\end{bmatrix}^\mathsf{T}.
\end{align}
Note that in both matrices, consecutive rows are identical, i.e. every pair of observers are correlated. In the second observation matrix, $\mxc_2$, the even rows are realizations of Gaussian vectors. As for the observation noise, we use a diagonal covariance $\mxr$ matrix parameterized by $\sigma^2_0=10^{-2}$ and $\sigma^2_1=1$. We use a bit string as a subscript for $\mxr$ to indicate the observers with low noise, i.e. a noise variance of $\sigma^2_0$, and those with high noise, i.e. a noise variance of $\sigma^2_1$. The objective behind using multiple observation and observation noise covariance matrices is identifying the effect of the observation model and noise on the choice of optimal observation sequence; for example, 
\begin{align}
\mxr_{010000} &=
\begin{bmatrix}
\sigma^2_0 & 0 & 0 & 0 & 0 & 0	\\
0 & \sigma^2_1 & 0 & 0 & 0 & 0	\\
0 & 0 & \sigma^2_0 & 0 & 0 & 0	\\
0 & 0 & 0 & \sigma^2_0 & 0 & 0	\\
0 & 0 & 0 & 0 & \sigma^2_0 & 0	\\
0 & 0 & 0 & 0 & 0 & \sigma^2_0	\\
\end{bmatrix}.
\end{align}

\paragraph{Estimate quality improves with observation rate} 
We verify the intuition that the frequency of observations improves the quality of the estimate. In \figref{fig:kf_1}, we plot the evolution of true system state with time as well as two state estimates obtained from observations produced at different rates: once every $0.003$ and once every $0.053$. When the observation period is short, estimation from incoherent observations tracks true state. When the observation period is long, estimation will be mostly prediction in cycles when no observations are available. For this particular experiment, we have used the observation matrix $\mxc_2$ and the observation noise matrix $\mxr_{000000}$.

\begin{figure*}[t]
	\centering
	\includegraphics[width=1\textwidth]{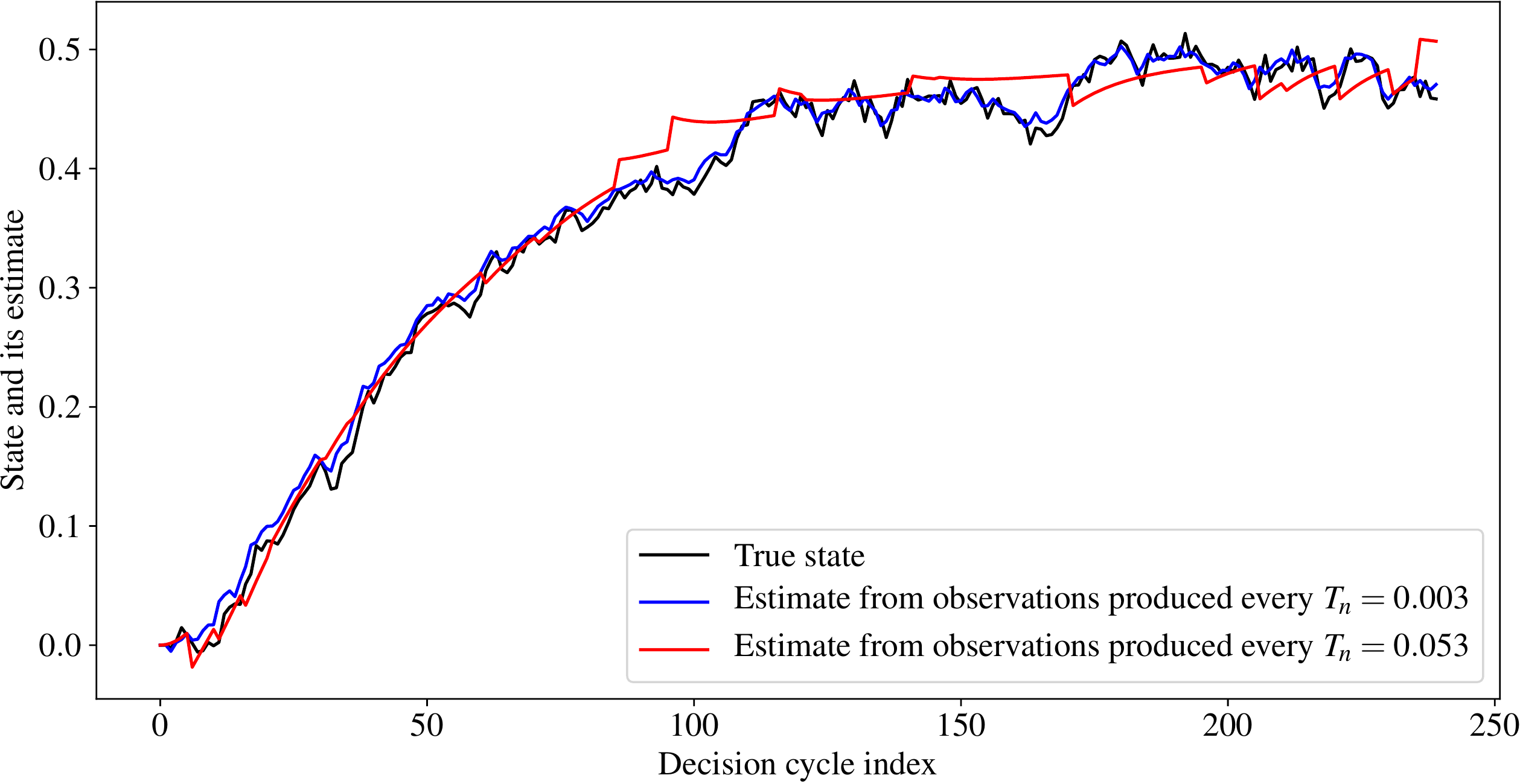}
	\caption{The true state (in black), its estimate from observations produced every $T_n=0.003$ (in blue), and the estimate from observations produced every $T_n=0.053$.}
	\label{fig:kf_1}
\end{figure*}

\paragraph{Performance spread between optimal and greedy algorithms depends on wireless channel quality} We plot the true system state, its estimates using optimal algorithm, and the estimate using the greedy algorithm; we consider two scenarios. In \figref{fig:baseline_diff}, the optimal algorithm outperforms the greedy algorithm. This occurs when early observations require long airtimes or when early observations are noisy. In \figref{fig:baseline_same}, the two algorithms perform similarly. This occurs when early observations require short airtimes or when early observations have little noise. For this particular experiment, we have used the observation matrix $\mxc_2$ and the observation noise matrix $\mxr_{000000}$.

\begin{figure*}[t]
	\centering
	\subfloat[Optimal outperforms greedy\label{fig:baseline_diff}]{
		\includegraphics[width=0.51\textwidth]{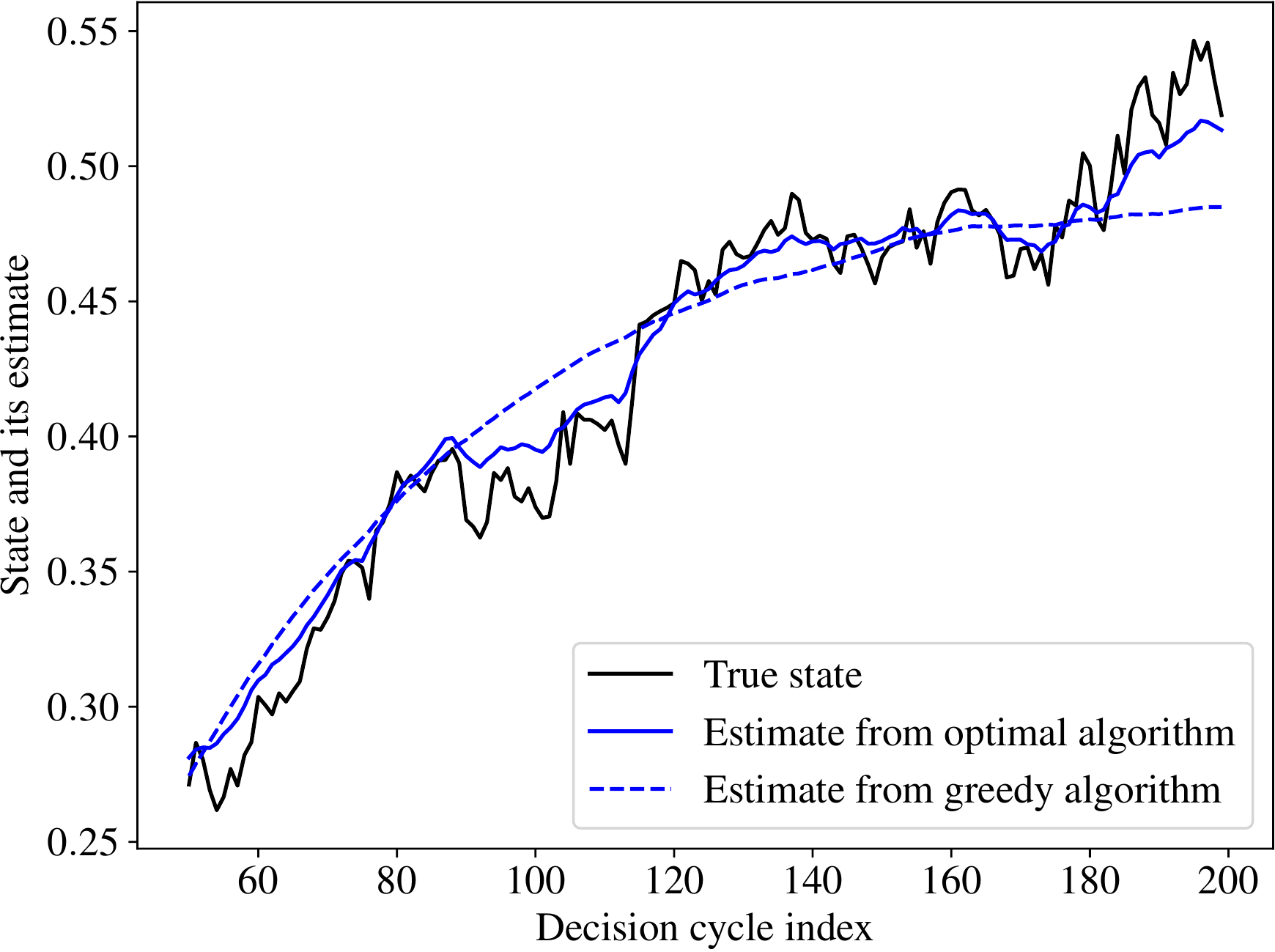}
	}
	\subfloat[Greedy compares to greedy\label{fig:baseline_same}]{
		\includegraphics[width=0.51\textwidth]{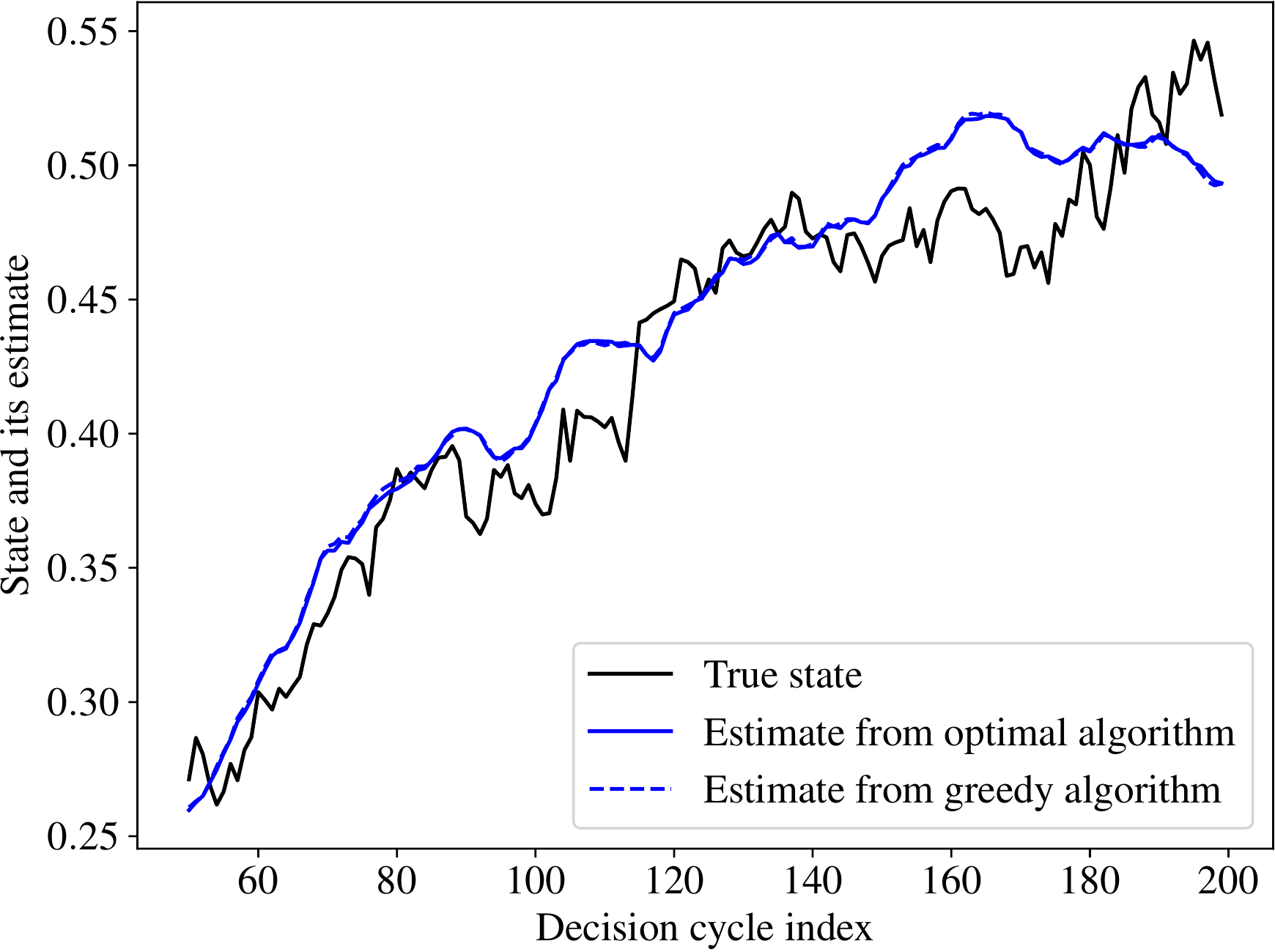}
	}
\caption{The true state (in blue) and its estimate obtained through the optimal algorithm (solid blue) and the greedy algorithm (dashed blue). In \figref{fig:baseline_diff}, the optimal algorithm outperforms the greedy algorithm. In \figref{fig:baseline_same}, the estimates of the two algorithms track the true state almost identically.}
\end{figure*}

\paragraph{Executive discards noisy observation when observability is limited}
We conduct an experiment with the following setup: 1) There are 6 observers, 2) the number of harvested observations to 4 out of 6, 3) every pair of observers are correlated, i.e. $\vc_0=\vc_1$, $\vc_2=\vc_3$ and $\vc_4=\vc_5$, and 4) 1 out of 6 observers is blacked-out, i.e. it has an observation noise variance of $\sigma_1^2$ while the rest have a noise variance of $\sigma_0^2 \ll \sigma_1^2$. Under this setup,   we should expect that the executive polls the observers with less noisy observations. Indeed, the pie charts of \figref{fig:noisy_is_bad_6of7} show that when 4 out of 6 observations can be harvested, the noisy observation is consistently discarded.

\begin{figure*}[t]
\centering
\includegraphics[width=1\textwidth]{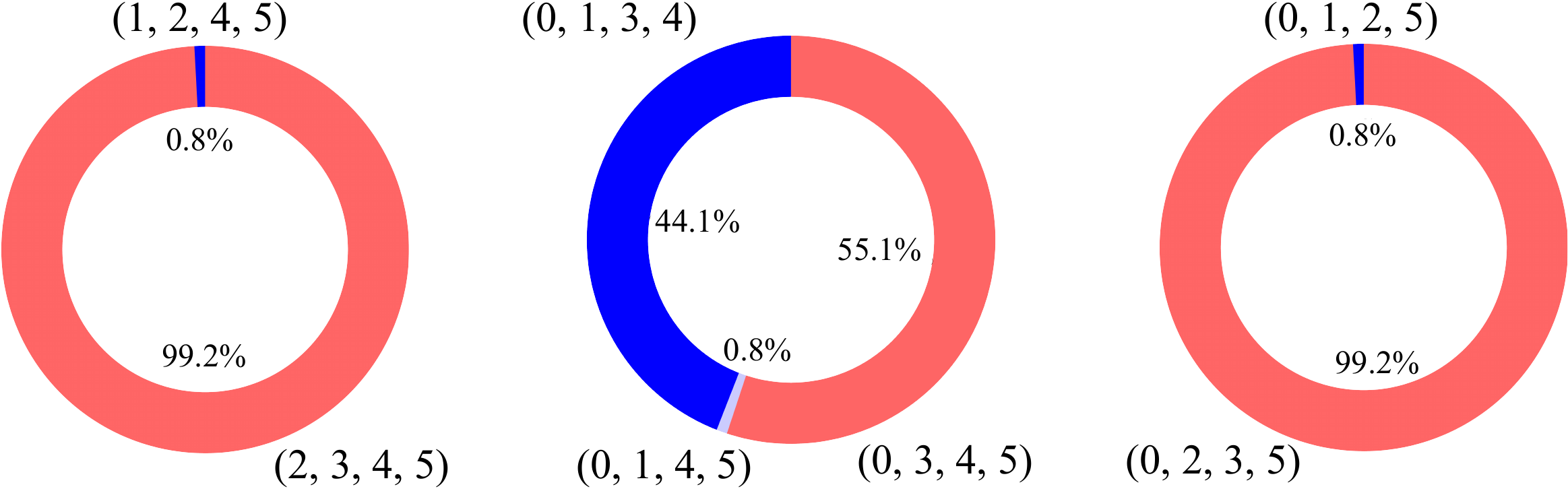}
\caption{When only 4 out of 6 observations can be harvested, the executive consistently excludes the most noisy one. For the pie chart on the left, observer $0$ is blacked out, i.e. the observation noise matrix $\mxr=\mxr_{100000}$. For the chart in the middle, observer $2$ is blacked out. For the chart on the right, observer $4$ is blacked out.}\label{fig:noisy_is_bad_6of7}
\vspace{-20pt}
\end{figure*}

\paragraph{When selection is unconstrained, all observations are harvested} We investigate whether the use of more observations to build the estimate reduces its error.
While this is reminiscent of the information theoretic dogma that information can't hurt \cite{cover:thomas}, \eqref{eq:objective_function} does not prove or disprove.
According to the pie chart \figref{fig:noisy_is_bad_6of7}, we observe that this is in fact true. When the selection is unconstrained, that is when there is ample time to harvest as many observations as needed, the executive selects all observations. In other words, the observation sequence that maximizes the executive's belief about the state of the system includes all observations. For this pie chart, we used $\mxc=\mxc_1$ and $\mxr=\mxr_{000000}$, but the results hold for an arbitrary choice of $\mxc$ and $\mxr$.

\section{Conclusion}\label{s:conc}
In this paper, we have developed a framework that abstracts away the context around different control and multistage decision processes by using a common mathematical model to formulate and solve OSP that the executive solves to schedule the observations that maximize its belief about the state of the system. First, we derived novel Kalman filter equations to predict the state of an \gls{lti} system at decision cycle boundaries from scalar observations produced at different resolutions. Second, we designed a \glsfirst{bnb} algorithm to optimally solve the observer selection problem that systematically iterates over observation sequences. Our numerical simulations showed that the executive selects all observations when there is no limit on the number of observations to be harvested, which suggests that more observations make for a more precise state estimate. With that, OSP bears resemblance to the \textit{knapsack problem}, whose objective is to select a number of items to maximize the total value while meeting a weight limit. While the knapsack problem is a combinatorial optimization problem, it has an efficient dynamic programming solution that could inspire a similar solution to OSP.
\appendices
\section{}\label{app:a}The derivation presented in this appendix mimics the discretization of continuous-time state-space equations in \cite[Chap. 5]{ogata:book}. The state of the system at time $t[l]$ is 
\begin{align}
\vx(t[l]) 
&= e^{\mxa t[l]} \vx(0) 
+ \int\limits_{0}^{t[l]} e^{\mxa(t[l]-\tau)} \mxb \vu(\tau) \diff\tau 
+ \int\limits_{0}^{t[l]} e^{\mxa(t[l]-\tau)} \vv(\tau) \diff\tau \\
\begin{split}
&= e^{\mxa(t[l]-t[k])} \left[ e^{\mxa t[k]} \vx(0) 
+ \int\limits_{0}^{t[k]} e^{\mxa(t[k]-\tau)} \mxb \vu(\tau) \diff\tau 
+ \int\limits_{0}^{t[k]} e^{\mxa(t[k]-\tau)} \vv(\tau) \diff\tau \right] \\
&\;\;\;\;\;\;\;\;\;\;\;
+ \int\limits_{t[k]}^{t[l]} e^{\mxa(t[l]-\tau)} \mxb \vu(\tau) \diff\tau 
+ \int\limits_{t[k]}^{t[l]} e^{\mxa(t[l]-\tau)} \vv(\tau) \diff\tau.
\end{split}\\
&= e^{\mxa(t[l]-t[k])} \vx(t[k]) 
+ \int\limits_{t[k]}^{t[l]} e^{\mxa(t[l]-\tau)} \mxb \vu(\tau) \diff\tau 
+ \int\limits_{t[k]}^{t[l]} e^{\mxa(t[l]-\tau)} \vv(\tau) \diff\tau. \label{eq:app_a_1}
\end{align}

Assuming, $\vu(t)=\vu[k]$ when $kT\leq t \leq (k+1)T$, the middle integral of \eqref{eq:app_a_1} becomes
\begin{align}
\int\limits_{t[k]}^{t[l]} e^{\mxa(t[l]-\tau)} \mxb \vu(\tau) \diff\tau 
&= \int\limits_{t[k]}^{(k+1)T} e^{\mxa(t[l]-\tau)} \mxb \vu[k] \diff\tau
+\sum_{j=k+1}^{l-1} \int\limits_{jT}^{(j+1)T} e^{\mxa(t[l]-\tau)} \mxb \vu[j] \diff\tau
+\int\limits_{lT}^{t[l]} e^{\mxa(t[l]-\tau)} \mxb \vu[l] \diff\tau,
\end{align}
which leads to the desired expression after a change of variables $s=t[l]-\tau$.

Assuming $\vv(t)\sim\mathcal{N}(\vzero,\mxq \delta(t))$, the third integral of \eqref{eq:app_a_1} will be distributed according to $\mathcal{N}(\vzero,\mxq(t[k],t[l]))$, where
\begin{align}
\mxq(t[k],t[l])
= \expectno \left[\int_{t[k]}^{t[l]} e^{\mxa(t[l]-\tau)} \vv(\tau) \diff\tau\right]^2
= \int\limits_{0}^{t[l]-t[k]} e^{\mxa s} \mxq e^{\mxa^* s} \diff s,
\end{align}
which is a standard result.

\section{}\label{app:b}First, if $T_n=T$, observer $n$ produces exactly one observation every decision cycle at its very beginning. Second, if $T_n > T$, then a cycle will contain either one observation from observer $n$ or none. To determine whether cycle $k$ has an observation, we compute the offset $\delta_n[k]$ between the $k$th decision cycle and the latest observation cycle. While $\floor{kT/T_n}$ specifies how many observation cycles have elapsed by the end of  decision cycle $k$, $\delta_n[k]$ specifies the remainder of this division. Accordingly, we express $\delta_n[k]$ as
\begin{equation}\label{eq:delta}
\delta_n[k] = kT \pmod{T_n} = kT - \floor{\frac{kT}{T_n}}\, T_n.
\end{equation}
If $\delta_n[k] \leq T$, then decision cycle $k$ contains exactly one observation at time $kT - \delta_n[k]$; otherwise, it contains none. The third and final case to consider is $T_n < T$. In this case, every decision cycle contains at least one observation. The last observation in decision cycle $k$ will be produced at time $kT-\delta_n[k]$. To determine the timestamp of the first observation, we solve the following equation for the largest integer $m$:
\begin{align}
kT-\delta_n[k] -mT_n \geq (k-1) T,
\end{align}
where the quantity of the left hand side of the inequality gives the timestamp of the first observation. Rearranging terms, we have
\begin{align}
m &\leq \frac{\delta_n[k]-T}{T_n}.
\intertext{Since we are solving for the largest $m$,}
m &= \floor{\frac{\delta_n[k]-T}{T_n}}.
\intertext{Writing $\delta_n[k]$ explicitly, and noting the idempotence of the floor operation,}
m &= \floor{\frac{kT}{T_n}} - \ceil{\frac{(k-1)T}{T_n}}.
\end{align}
The first observation is thus
\begin{align}
\begin{split}
\myt{n}{0}{k} &= kT - \delta_n[k] -mT_n \\
&= \ceil{\frac{(k-1)T}{T_n}} T_n.
\end{split}
\end{align}
\bibliographystyle{IEEEtran}
\bibliography{IEEEabrv,ref}
\end{document}